\documentclass[journal]{IEEEtran}

\usepackage{algorithm}
\usepackage{algorithmic}
\usepackage{amsmath}
\usepackage{amssymb}
\usepackage{bm}
\usepackage{bbm}
\usepackage{graphicx}
\usepackage{xcolor}
\usepackage{amsthm}
\usepackage{subfigure}

\usepackage{tcolorbox}

\newtheorem*{Question}{Question}

\newtheorem{Theorem}{Theorem}

\newtheorem{Remark}{Remark}

\def\N{\mathcal{N}}

\def\S{\mathcal{S}}

\def\S{\mathcal{S}}

\def\bx{\mathbf{x}}
\def\by{\bm{y}}
\def\bg{\mathbf{g}}

\def\hy{\hat{y}}

\def\bomega{\bm{\omega}}

\def\btheta{\mathcal{\bm{\theta}}}

\usepackage[utf8]{inputenc}

\newcommand{\Meng}[1]{\ifthenelse{\boolean{showcomments}}
{ \textcolor{blue}{(Meng says:  #1)}}{}}

\ifodd 1

\newcommand{\com}[1]{\textbf{\color{red} (Comment: #1) }}
\newcommand{\comg}[1]{\textbf{\color{blue} (COMMENT: #1})}
\newcommand{\response}[1]{\textbf{\color{blue} (RESPONSE: #1})}
\newcommand{\deleting}[1]{}
\else

\newcommand{\com}[1]{}
\newcommand{\comg}[1]{}
\newcommand{\response}[1]{}
\newcommand{\deleting}[1]{}
\fi

\title{MoE$^2$: Optimizing Collaborative Inference\\ for Edge Large Language Models}

\author{
 Lyudong Jin, Yanning Zhang, Yanhan Li, Shurong Wang, Howard H. Yang,~\IEEEmembership{Member,~IEEE}, Jian Wu,~\IEEEmembership{Senior Member,~IEEE}, and Meng Zhang,~\IEEEmembership{Member,~IEEE}
     \thanks{Lyudong Jin, Yanning Zhang, Yanhan Li, Shurong Wang, Howard H. Yang, and Meng Zhang are with the Zhejiang University\text{—}University of Illinois at Urbana\text{–}Champaign Institute, Zhejiang University, Haining 314400, China (e\text{-}mails: 3180101183@zju.edu.cn;  yanning.22@intl.zju.edu.cn;  yanhan.24@intl.zju.edu.cn;  shurong.22@intl.zju.edu.cn; haoyang@intl.zju.edu.cn;   mengzhang@intl.zju.edu.cn). 
     
     Jian Wu is with Zhejiang Key Laboratory of Medical Imaging Artificial Intelligence, Zhejiang University, Hangzhou, China (e\text{-}mail: wujian2000@zju.edu.cn).}
    }
\begin{document}

\maketitle

\begin{abstract}

Large language models (LLMs) have demonstrated remarkable capabilities across a wide range of natural language processing tasks. 
% On-device large language models (LLMs), referring to running LLMs on edge devices, have raised considerable interest since they are more cost-effective, latency-efficient, and privacy-preserving compared with the cloud paradigm.
Exploiting the heterogeneous capabilities of edge LLMs is crucial for diverse emerging applications, as it enables greater cost-effectiveness and reduced latency.
% However, they pose significant challenges due to the heterogeneity of edge LLMs across various attributes, such as model capability, delay, and energy consumption.
% challenges due to their high energy consumption and latency constraints. 
In this work, we introduce \textit{Mixture-of-Edge-Experts (MoE$^2$)}, a novel collaborative inference framework for edge LLMs. 
 We formulate the joint gating and expert selection problem to optimize inference performance under energy and latency constraints. 
Unlike conventional MoE problems, LLM expert selection is significantly more challenging due to the combinatorial nature and the heterogeneity of edge LLMs across various attributes. 
To this end, we propose a two-level expert selection mechanism through which we uncover an optimality-preserving property of gating parameters across expert selections. This property enables the decomposition of the training and selection processes, significantly reducing complexity.
Furthermore, we leverage the objective's monotonicity and design a discrete monotonic optimization algorithm for optimal expert selection. 
% By leveraging the sparsity and modularity of MoE, our approach achieves a favorable trade-off between performance, latency, and energy efficiency, making it well-suited for real-time applications in mobile edge computing scenarios.
We implement edge servers with NVIDIA Jetson AGX Orins and NVIDIA RTX 4090 GPUs, and perform extensive experiments. Our results validate that performance improvements of various LLM models and
show that our MoE$^2$ method can achieve optimal trade-offs among different delay and energy budgets, and outperforms baselines under various system resource constraints.

\end{abstract}

 
\section{Introduction}

\subsection{Background and Motivations}
\textit{Large language models (LLMs)} represent a significant breakthrough in artificial intelligence, particularly in the field of natural language processing (NLP).
Built primarily on the transformer architecture \cite{transformer}, LLMs are trained on vast and diverse datasets, enabling them to generate and contextualize human language to a remarkable extent.
Prominent examples, such as OpenAI's GPT series \cite{DBLP:journals/corr/abs-2303-08774, DBLP:conf/nips/Ouyang0JAWMZASR22} and Google's PaLM \cite{DBLP:journals/jmlr/ChowdheryNDBMRBCSGSSTMRBTSPRDHPBAI23, DBLP:conf/icml/DriessXSLCIWTVY23, anil2023palm}, demonstrate the capability of LLMs to perform a wide range of tasks, including text generation, summarization, translation, and question answering. 
Their versatility stems from their ability to learn and generalize patterns in language, making them indispensable for applications spanning from conversational AI to code generation and beyond. 
However, this capability comes at the cost of immense computational and memory requirements, as LLMs often contain billions of parameters. 
This resource-intensive nature creates challenges for their deployment, particularly in environments where computational resources are limited, such as edge devices.
As the demand for real-time, localized AI services grows, addressing these challenges has become a critical focus for researchers and practitioners alike.

\begin{figure}[!t]
    \centering
    \includegraphics[width=0.9\linewidth]{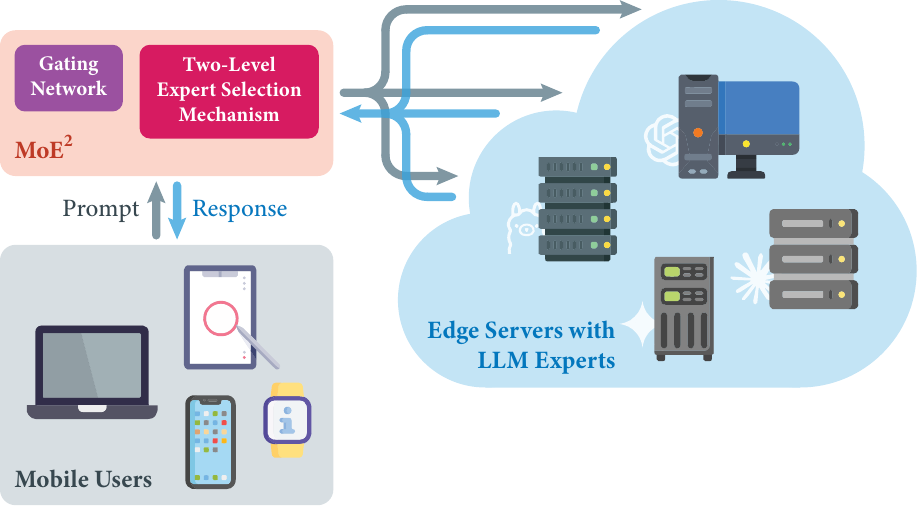}
    \caption{An overview of the MoE$^2$ framework, which leverages a gating network and a two-level expert selection mechanism for efficient task processing. Mobile users send prompts, which are routed to a subset of edge servers with LLM experts. The expert selection method is carefully designed to enable a scalable and efficient framework for diverse LLM applications.}
    \label{fig:system}
    \vspace{-1\baselineskip}
\end{figure}

On the other hand, 
% mobile edge computing (MEC) networks, mobile users randomly arrive over time to offload intensive machine learning tasks with unknown data distributions to edge servers with superior computing capabilities (e.g., [1], [2])
the convergence of edge intelligence, mobile edge computing, and LLMs is redefining the landscape of AI-driven applications by enabling powerful computational capabilities closer to end users. \textit{Mobile Edge Computing (MEC)} is a paradigm that brings computational resources closer to end users by deploying processing power and storage at the edge of mobile networks, while edge intelligence leverages the processing power of edge devices to perform real-time, localized decision-making. 
% Large language models, with their ability to generate human-like text and perform a wide range of natural language processing tasks, are increasingly being integrated into edge environments to provide personalized, context-aware services. 
However, the deployment of LLMs at the edge is inherently challenging due to their high computational and memory requirements, which often exceed the resource constraints of edge devices. This interplay between edge intelligence, mobile edge computing, and LLMs highlights the need for innovative task offloading techniques to balance performance, resource efficiency, and user experience. Understanding this dynamic is crucial for unlocking the full potential of LLMs in edge settings, enabling applications such as real-time translation, conversational AI, and intelligent assistance to flourish in resource-constrained environments.

In this paper, we aim to bridge this gap by answering the following key question:
\begin{Question}
 How can collaborative inference leverage heterogeneous edge LLMs  to handle diverse task requests (prompts) to  balancing performance, latency, and energy consumption?
    % considering    
    % an edge LLM system be designed to balance performance, latency, and energy efficiency?
\end{Question}
    
\subsection{Key Challenges and Solution Approach}

In this work, we propose deploying \textit{Mixture-of-Experts (MoE)} paradigm at network edges to enable efficient collaborative edge inference. Originally introduced to tackle complex tasks by dividing them into subtasks and assigning them to specialized experts, MoE \cite{jacobs1991adaptive, jordan1994hierarchical} has evolved into a powerful approach for efficiently up-scaling models through sparse activation \cite{DBLP:conf/iclr/ShazeerMMDLHD17, DBLP:journals/corr/abs-2006-16668, fedus2022switch}.
In this sense, MoE is particularly suited for addressing the computational and efficiency challenges posed by large-scale models, such as LLMs.
Unlike traditional models, which activate all parameters during inference, MoE dynamically activates a small subset of specialized ``expert'' submodels for each task \cite{jacobs1991adaptive}.
This dynamic routing mechanism significantly scales model capacity without a proportional increase in computational overhead, ensuring high efficiency for resource-intensive tasks.
By leveraging only the most relevant experts for a given task, MoE not only achieves better performance but also optimizes resource utilization, such as memory and processing power.
MoE's modular and sparse traits has great potential in reducing training and inference costs for LLMs.

In the context of edge networks, MoE provides a powerful approach for leveraging the heterogeneous capabilities of edge LLM experts by adapting computations to local resource constraints and task-specific requirements. By selecting a subset of edge LLM experts tailored to each inference task, MoE effectively minimizes computational overhead and latency. Furthermore, strategically deploying MoE at the edge enables optimization of LLM performance while reducing energy consumption and latency, thereby enhancing the overall user experience.

However, several key challenges remain, including:
\begin{enumerate}
    \item \textbf{Various Attributes across Edges:}
    % Compared to conventional MoE architectures, lies in the joint selection of LLM experts and the training of gating parameters. 
    Traditional MoE approach typically choose experts merely based on gating values.
    % These routing networks often add tunable noise and retain only the top $k$ values before applying the softmax function (as in \eqref{softmax}).
    Although this approach may reduce computational overhead, it cannot be deployed directly in networks with edge LLMs that are heterogeneous in various attributes, including capability, energy consumption, and latency.
    % The deployment of LLMs on edge devices requires significant computational resources, leading to high energy consumption and latency, which compromises the system efficiency and user experience. In addition, LLMs can have billions of parameters, resulting in models that are several gigabytes in size. That makes it challenging to deploy them on edge devices with limited memory and storage capacity without compromising performance.
    \item \textbf{Combinatorial Optimization:} To satisfy the system constraints, we proposed an optimization problem to select the optimal subset of LLM experts for each query. However, this problem is challenging due to the expert subset selection's combinatorial nature, the loss function's non-convexity, and the complex interplay between system constraints. Traditional optimization techniques may not be suitable to resolve these competing objectives simultaneously.
\end{enumerate}

In this work, we introduce \textit{Mixture-of-Edge-Experts (MoE$^2$)}, a novel collaborative inference framework for edge LLMs, as shown in Fig. \ref{fig:system}. Our framework introduces a two-level expert selection mechanism. At the coarse-grained level, the selection is optimization-based, ensuring worst-case bounds on energy consumption and latency. At the fine-grained level, experts are dynamically selected based on input prompts through a routing network. 
This approach effectively leverages the heterogeneity in the capabilities of edge LLM experts to handle diverse tasks.
Our main contributions are summarized as follows:
\begin{itemize}
    \item \emph{MoE Design and Problem Formulation:}
    % We design the MoE framework for   MoE-aided inference over the edge LLMs \textcolor{red}{X}.
    We formulate the first MoE-aided inference offloading problem for edge LLMs by optimally designing the \textit{gating network} and \textit{expert selection}, subject to energy consumption and latency constraints. This problem is inherently combinatorial and challenging due to the heterogeneous nature of edge LLMs.
    \item \emph{Optimal Solution Structures:} To address the challenges of combinatorial optimization, we leverage the structural properties of the optimal solution. Specifically, we show that: (i) the optimality of gating parameters for the full LLM set extends to its subsets, and (ii) the objective value exhibits a monotonic improvement property during LLM set selection. These analytical insights demonstrate that the training of gating parameters can be decoupled from the selection of LLM experts, and facilitate our optimization algorithm design.
    \item \emph{Algorithm Design.} The MoE$^2$ framework, including a training stage and an inference stage, is designed based on two derived theorems. 
    % which explore the optimality and monotonic improvement property of subset selection.
    We also propose a cluster-based domain identification method to create domain experts, which efficiently utilizes the strength of our MoE framework.
    \item \emph{Implementation:} We implement MoE$^2$ on edge servers, and perform extensive experiments. Our results validate that performance improvements of various LLM models and show that MoE$^2$ can adapt to different delay and energy budgets with optimal trade-offs, and achieves strong performance compared to baselines under various system resource constraints.
\end{itemize}

We organize the rest of this paper as follows. Section \ref{Sec:Literature} presents the literature review. Section \ref{Sec:Model}
presents the system model and formulates the problem. The proposed scheme is presented in Section \ref{Sec: Methodology}. Sections \ref{Sec: Simulation Results} and \ref{Sec:Experiment} present the simulation and experimental results, respectively. Finally, this paper is concluded in Section \ref{Sec: Con}.

\section{Related Works} \label{Sec:Literature}

\subsection{Mobile Edge Computing and Collaborative Inference}

In recent years, extensive research has explored MEC across various settings (e.g., \cite{ouyang2019adaptive, tran2018joint, singh2017optimize, mao2017survey}). This work primarily reviews the collaborative inference literature in this domain, which focuses on deploying machine learning models at the edge to enable real-time data processing and decision-making \cite{li2019edge}.
The collaborative inference literature can be categorized into two main approaches: \textit{model partitioning} and \textit{edge-cloud collaboration}.
In edge-cloud collaboration, edge devices and the cloud cooperate to balance system performance with resource constraints.
Time-sensitive tasks are processed locally on edge devices, while more complex or less time-critical workloads are uploaded to the cloud for additional computational support (e.g., \cite{lin2019computation, mach2017mobile, he2024large}).
This cooperative mechanism has been leveraged to accelerate inference for large language models (LLMs) \cite{he2024large}.
To optimize edge computing, various algorithms have been developed to enable collaborative inference across multiple edge devices, effectively utilizing their combined computational capabilities to enhance model performance (e.g., \cite{li2022collaborative, hu2022distributed, shi2021dnn, liu2022hastening}). 
Integrating edge computing with AI and edge-cloud collaboration is highly promising for Internet of Things IoT applications requiring low latency, high bandwidth efficiency, and real-time processing \cite{10.1109/mnet.2018.1700202}.
\textit{Our proposed MoE$^2$ introduces a novel framework for collaborative inference by leveraging the diverse capabilities of LLM experts specialized in handling different tasks, distinguishing it from existing approaches.}\par

\subsection{Mixture-of-Experts (MoE)}
% the concept of MoE & the traditional use of MoE
The concept of MoE was first introduced in \cite{jacobs1991adaptive, jordan1994hierarchical} to dynamically assign inputs to \textit{multiple expert networks} using \textit{a gating network}. Early MoE models improved adaptability and accuracy by dividing tasks into sub-tasks, with each expert specializing in specific regions of the input space. The gating network dynamically selects the most relevant experts, enabling efficient problem-solving.

Modern MoE architectures (e.g., \cite{DBLP:conf/iclr/ShazeerMMDLHD17, DBLP:journals/corr/abs-2006-16668, fedus2022switch,li2024locmoe}), introduced \textit{sparse activation}, where the gating network activates only a subset of experts for each input, enabling models to scale efficiently to trillions of parameters. Building on this foundation, MoE has recently been applied to address the issue of performance degradation when deploying LLMs on edge devices (e.g., \cite{yi2023edgemoe,  shen2024jetmoe}). \textit{Our work extends these principles by designing an edge-based MoE system to balance performance, latency, and energy efficiency.}
% Smaller LLMs act as domain-specific experts deployed on edge servers, with a gating network routing inputs dynamically to ensure efficient resource use and real-time processing.}

\subsection{Edge LLMs}
% many works delve into the combination of llms and edge devices
The superior performance of LLMs has increasingly captured researchers' attention, leading to a growing focus on integrating LLMs with edge devices, as recently surveyed in \cite{10835069}.

% on-device llms(edge4llms)
\textbf{Edges/Networks for LLMs:}
Edge deployment of LLMs offers solutions to the challenges of cloud-based systems, such as high latency, data security concerns, and connectivity limitations. By leveraging edge computing, LLMs can process data locally, improving response times, enhancing privacy, and reducing bandwidth usage, which makes them particularly suitable for latency-sensitive and resource-constrained applications \cite{10835069}. Recent advancements in edge inference techniques for LLMs can be categorized into centralized edge inference, split inference, and collaborative inference. Centralized edge inference reduces communication overhead by optimizing token representations (e.g., \cite{liang2022not, jiang2023llmlingua}) and employing methods such as service placement and migration (e.g., \cite{xu2024cached, ding2024hybrid, fu2024serverlessllm, fang2023large}). Split inference divides tasks between edge devices and servers to balance computation and communication, using strategies like token representation reduction (e.g., \cite{goyal2020power, shao2021learning, chen2023cross}), progressive offloading \cite{lan2022progressive}, early exit \cite{schuster2022confident} and multi-hop architectures \cite{ma2023poster}. Collaborative inference employs speculative decoding (e.g., \cite{leviathan2023fast, wang2023tabi}), where lightweight models on edge devices generate initial predictions that are refined by powerful servers. These techniques collectively enable scalable and resource-efficient edge inference for LLMs. \textit{However, when deploying LLMs at the edge, the trade-offs between performance, latency, and energy efficiency are not always fully considered. In addition, our approach did not require intrusive modifications, i.e., any modifications to edge LLM experts.}

% think about our advantages: 1) multi small LLMs (equal to one large llm?); 2) without intrusive modification

% llms for edge devices(llms4edge)
\textbf{LLMs for Edges/Networks:}
LLMs are transformative in wireless communications, addressing tasks like network configuration, traffic classification, and 6G optimization by reducing human effort and improving efficiency \cite{zhou2024large}. They are increasingly used for telecom-related tasks, such as domain knowledge generation (e.g., \cite{soman2023observations, wang2024grammar}), code generation (e.g., \cite{du2023power, mani2023enhancing, xiang2023toward, thakur2023benchmarking}), and network configuration generation (e.g., \cite{dzeparoska2023llm, wang2023making, mondal2023llms}). LLMs also excel in classification tasks, including network security (e.g., \cite{aghaei2022securebert, ameri2021cybert, yin2020apply, ferrag2024revolutionizing, seyyar2022attack}), text (e.g., \cite{DBLP:conf/globecom/BariahZZMBD23, aftan2023using}), image (e.g., \cite{menon2022visualclassificationdescriptionlarge, DBLP:conf/iccv/PrattCLF23}), and network traffic classification (e.g., \cite{shi2023bfcn, DBLP:conf/www/LinXGLSY22}). In network optimization, LLM-enabled techniques like reinforcement learning (e.g., \cite{DBLP:journals/corr/abs-2309-06687, kwon2023reward, DBLP:conf/iclr/MaLWHBJZFA24}), black-box optimization \cite{guo2023towards}, convex optimization (e.g., \cite{DBLP:journals/corr/abs-2308-12923, DBLP:journals/corr/abs-2310-06116}), and heuristic algorithms (e.g., \cite{DBLP:conf/gecco/PluhacekKKVS23, DBLP:journals/corr/abs-2310-12541}) enhance wireless network management. Additionally, LLMs apply time series models like pre-trained foundation models (e.g., \cite{DBLP:journals/corr/abs-2310-03589, DBLP:conf/icml/DasKSZ24}), frozen pre-trained models (e.g., \cite{DBLP:journals/tkde/XueS24, DBLP:conf/iclr/0005WMCZSCLLPW24}), fine-tuning (e.g., \cite{chang2024llm4ts, DBLP:conf/nips/ZhouNW0023}), and multi-modality approaches \cite{DBLP:journals/corr/abs-2307-10802} to predict trends and demands. Wu \textit{et al.} further pioneer LLMs as foundational models for networking with their NetLLM framework, reducing manual efforts and improving adaptability \cite{wu2024netllm}.
\textit{Our study fundamentally differs from the aforementioned works on ``LLMs for networks'', where the primary objective is to utilize LLMs to optimize edge networks, whereas our focus is on leveraging edge resources to support LLMs.} \par

\textbf{Closely Related Works:} Only a couple of studies are closely related.
Yi \textit{et al.} \cite{yi2023edgemoe} proposed a framework which enables efficient on-device inference with a single LLM on a single edge server, with experts stored in external memory and loaded as needed. This work did not consider collaborative inference by exploiting heterogeneous LLM experts.
% However, the reloading of experts requires external memory and brings extra delay. 
% EdgeShard \cite{zhang2024edgeshard} introduced a method to partition a computation-intensive LLM into affordable shards and deploy them on distributed devices. 
% However, it brings higher latency and energy consumption due to distributed model sharding.
% \textit{In addition, it requires intrusive modifications to the LLMs.} 
A concurrent study by Li \textit{et al.} \cite{li2024theory} proposed a MoE model over edges for continual learning, while it did not consider the deployment of large models under a resource-constrained setting. 

\begin{figure*}[t]
    \centering
    \includegraphics[width=0.8\linewidth]{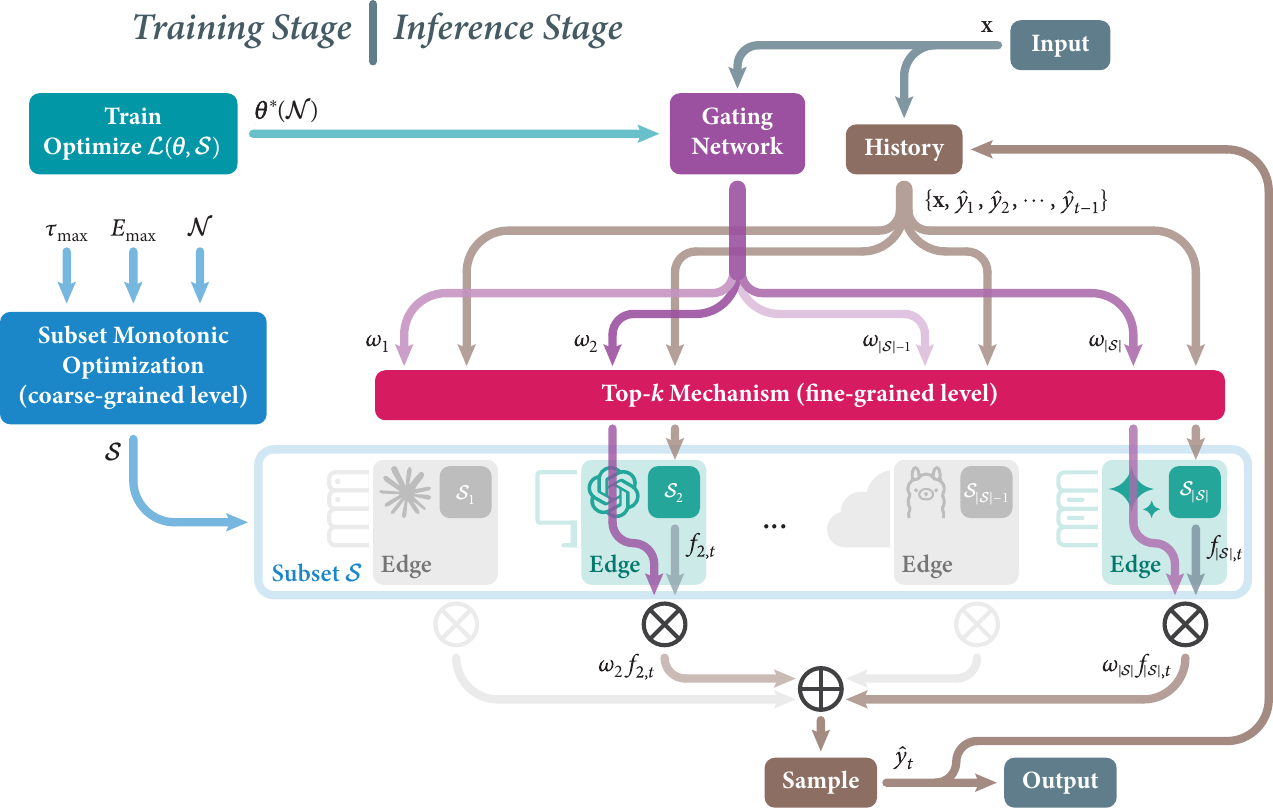}
    \caption{A detailed illustration of the MoE$^2$ framework, which consists of training stage and inference stage. In the training stage, the gating network parameters $\bm{\theta}$ are derived by optimizing $\mathcal L(\boldsymbol \theta, \S)$ over all LLM $\mathcal N$. Then, the subset selection $\S$ are obtained by Subset Monotonic Optimization algorithm to satisfy system constraints. In the inference stage, the gating network computes the gating values for each LLM experts, and Top-$k$ mechanism selects a smaller subset from $\S$ of $k$ LLMs with the highest gating values for responding to prompt $\bx$.}
    \label{fig:framework}
\end{figure*}

\section{System Model and Problem Formulation} \label{Sec:Model}

In this section, we present the system model for edge LLMs, the MoE$^2$ architecture, and the two-level expert selection scheme. We then formulate the joint gating parameter training and LLM expert selection problem.

\subsection{Edge LLM System Overview}

% \edit{We first \rev{introduce} the device and edge server model.}

\textbf{System Overview}. As shown in Fig. \ref{fig:framework}, we consider a system of a set $\mathcal{N}$ of specialized LLM experts (agents). Each LLM expert $n\in\mathcal{N}$ is deployed on one edge server.
% $e_n\in\E$.
 % Each LLM expert $n$ is 
 
\textbf{Task Model.} Mobile users with inference tasks may choose to offload them to edge LLM experts. 
Let $\mathcal{X}$ represent the set of all prompts requested by different users. Users may send prompts $\bx \in \mathcal{X}$ to edge servers based on their specific applications. 
 Future edge LLMs have diverse applications, including mobile health, humanoid robots, virtual assistants, and autonomous driving \cite{10835069}. 
 %  healthcare is one of the most promising applications for LLMs. 
For example, Google’s Med-PaLM 2 is an LLM fine-tuned on medical datasets, designed to provide answers to medical inquiries. Beyond this, healthcare LLMs can assist with tasks such as medical question answering, diagnosis, treatment planning, and medical report generation.

 We consider a set $\mathcal{M}$ representing categories of LLM applications, which constitute $M$ partitions of $\mathcal{X}$. Specifically, we use $\mathcal{X}_m$ to denote the set of prompts (tasks) $\bx$ associated with LLM application $m$. Each LLM application has its own latency requirements, which will be specified later.

\textbf{Autoregressive Process}. The system processes $\bm{x}$ and returns $T$ prediction tokens through the autoregressive process. 
 The length $T$ is decided by the experts through the autoregressive process defined as follows. Starting with the prompt $\bx$, an LLM agent generates output tokens $\hy_t$ at time step $t\in T$ sequentially. The system maintains a history $h(\bx,t)=\{\bx, \hy_1,\cdots,\hy_{t-1}\}$ that combines the original prompt and all previously generated tokens. 
 At time step $t$, an LLM expert $n \in \N$ may processes this history through the LLM output function $f_n$ to generate a probability distribution $f_n(h(\bx,t))\in \mathcal{P}(\mathcal{V})$ for the next token $\hy_{t}$, where $\mathcal{V}$ is the used vocabulary space. The generation of $\hy_{t}$ from $f_n(h(\bx,t))$ will be discussed later.

\textbf{Distributions}. Each prompt $\bx$ is paired with a target answer $\by = [y_t]_{t=1}^T$.
We let $\mathcal{P}_{\bx}$ be the probability distribution of different prompts, requested by the users, and let 
% Let $\mathcal{Y}(\mathcal{X})$ denote the target distribution of $\mathcal{X}$ and
$\mathcal{P}_{\bx,\by}$ denote the joint distribution of $(\bx,\by)$. For each LLM application $m \in \mathcal{M}$, we use $\mathcal{P}_{\bx,m}$ to denote the probability distribution of $\bx$ conditioned on $\bx \in \mathcal{X}_m$.

\subsection{The Mixture-of-Edge-Experts (MoE$^2$) Architecture}

\textbf{MoE$^2$}. The system employs an MoE-like architecture to select a subset of edge LLM experts to generate token predictions dynamically. 
MoE$^2$ consists of two key components: a \textit{gating network} and a \textit{two-level expert selection mechanism}.
% The gating network will be introduced later.
Given a prompt $\bx$, the generated distributions $f_n(h(\bx,t))$ are weighed by the gating network and aggregated at the user end to generate the final token $\hy_t$, which is subsequently added to the history $h(\bx,t)$ for future generation. This iterative process continues, with each new token generation benefiting from the accumulated context of all previous predictions.

% \textbf{LLM Applications:} \textcolor{red}{We will present four killer LLM-powered applications: . }

\subsubsection{Gating Network} Our system implements a distributed architecture where each edge server hosts a dedicated gating network that optimizes the routing weights for the accessible LLM experts. The gating network employs an embedding model to extract textual information and a multi-layer perceptron (MLP) to generate gating values for the LLM experts. Upon  receiving a prompt from a mobile user, the nearest edge server processes the request with its gating network, coordinates the token generation process, and synthesizes the final response for transmission back to the user. Let $\bm{\theta}$ denote the \textit{gating parameters} of the gating network. Given prompt $\mathbf{x}$, the gating network \cite{DBLP:conf/iclr/ShazeerMMDLHD17} outputs: 
\begin{equation}
    \mathbf{g}(\bx,\boldsymbol{\theta}) = [g_n(\bx,\boldsymbol{\theta})]_{n\in\mathcal{N}},
\end{equation}
where $g_n(\bx,\boldsymbol{\theta})$ is the $n$-th element of $\mathbf{g} (\bx,\boldsymbol{\theta})$. 
% For any 
An example gating function $g_n(\bx, \bm{\theta})$ is
% $\bm{\theta}$,  is guaranteed to be a real value between $0$ and $1$, such as when using 
the multilayer perceptron (MLP).
% \footnote{The sigmoid function is defined as $\sigma(x)={1}/({1+e^{-x}})$.}.
Given the gating values, we compute the weights of experts with different expert selections. Given prompt $\bx$, for any select subset of experts $\mathcal{S}\subseteq\mathcal{N}$, we compute the \textit{normalized weights} of $\mathcal{S}$ as:
\begin{equation}
    \bm{\omega}(\bx,\boldsymbol{\theta},\S) 
    =\left[\omega_n(\bx, \bm{\theta}, \S)\right]_{n\in\mathcal{S}}
    = \left[\frac{g_n(\bx,\boldsymbol{\theta})}{\sum\limits_{n'\in \mathcal{S}}g_{n'}(\bx,\boldsymbol{\theta})}\right]_{n\in\mathcal{S}}, \label{softmax}
\end{equation}
% where $\omega_n(\bx, \bm{\theta})$ is the $n$-th element of $\bm{\omega}(\bx,\boldsymbol{\theta},{\S})$.
Then, we can obtain the ensemble prediction by a weighted combination of expert outputs \cite{DBLP:conf/iclr/ShazeerMMDLHD17}. The fused probability distribution of token $k$ is expressed as:
\begin{equation}
    F(\bx, \boldsymbol{\theta}, \mathcal{S}, t) = \sum\limits_{n\in \mathcal{S}} \omega_n(\bx, \bm{\theta},\S) f_n(h(\bx,t)).
\end{equation}
Note that $F(\bx, \boldsymbol{\theta}, \mathcal{S}, t)$ is a probability distribution of $\mathcal{V}$, i.e., $F(\bx, \boldsymbol{\theta}, \mathcal{S}, t) \in \mathcal{P}(\mathcal{V})$. A new token $\hy_{t+1}$ is sampled directly from $F(\bx, \boldsymbol{\theta}, \mathcal{S}, t)$, i.e., $\hy_{t+1} \sim F(\bx, \boldsymbol{\theta}, \mathcal{S}, t)$. Then, $\hy_{t+1}$ will be added to the history information $h(\bx, t+1)$. This process repeats until a stop token is generated.

\subsubsection{Two-Level Expert Selection Mechanism}
One of the key challenges, compared to conventional MoE architectures, lies in the joint selection of LLM experts and the training of gating parameters. Traditional MoEs typically choose experts merely based on gating values,
employ additional routing networks to introduce sparsity. 
These routing networks often add tunable noise and retain only the top $k$ values before applying the softmax function (as in \eqref{softmax}). Although this approach can reduce computational overhead, it cannot be directly deployed in networks with edge LLMs that are heterogeneous across various attributes, including capability, energy consumption, and latency.

To address the above challenge, we propose an LLM expert selection scheme comprising two levels: a \textit{coarse-grained} level and a \textit{fine-grained} level, as illustrated in Fig. \ref{fig:framework}:
\begin{itemize}
    \item The coarse-grained level selects a subset of LLM experts, denoted by $\mathcal{S}$, based on system constraints, such as through optimization-based methods. This selection ensures a worst-case bound on energy consumption and induced latency, accounting for scenarios where all LLM experts in $\mathcal{S}$ are utilized. 
    \item In the fine-grained level, LLM experts are further selected from the subset $\mathcal{S}$ dynamically based on prompts $\bx$ via a routing network, leveraging the heterogeneity of LLM agents in handling different tasks while further reducing system costs. 
\end{itemize}
% These two 

As we will demonstrate later, the two-level expert selection mechanism, combined with a sufficiently sophisticated gating function $\mathbf{g}(\bx, \bm{\theta})$, is crucial for enabling efficient algorithm design to jointly train $\bm{\theta}$ and optimize $\mathcal{S}$.

\subsubsection{Where to Deploy/Train MoE$^2$} Based on our experimental results presented later, a gating network with superior performance can be as small as 1.7GB, enabling its deployment on edge servers (alongside edge LLMs) to serve nearby users effectively. This also opens up the possibility of adopting edge federated learning or cloud-edge collaboration frameworks, where training is conducted in the cloud and deployment occurs at edges.

\subsection{System Constraints}

In reality, LLM applications typically have specific service requirements, which make deploying LLMs at the mobile edge a more efficient solution.
The system incurs costs from two primary sources: system delay and energy consumption.
% With
% These LLM applications, may have the following service requirements, making it better to deploy LLMs at the mobile edge.
% The system incurs costs from two sources: the system delay and the energy consumption of the system.

\subsubsection{System Delay}

The system delay in our system can be categorized into three components: the computational time of the LLM experts, the transmission latency between edge devices and mobile devices, and the ensemble time when aggregating the outputs of LLM experts. For an edge LLM agent $n$, the computation time of prompt $\bx$ can be represented as:
\begin{equation}
    \tau_{n}^{\mathrm{comp}}(\bx) = \frac{C^{\mathrm{token}}_{n}(\bx)}{C_{n}^{\mathrm{cap}}} + \frac{M_{\mathrm{size}}}{b_{n}} + C^{\mathrm{ovh}}_{n},
\end{equation}
where:
\begin{itemize}
    \item $C^{\mathrm{token}}_{n}(\bx)$: Average computational cost per token (in FLOPs) for edge LLM model $n$ for input $\bx$.
    %  (in FLOPs) 
    \item $C_{e}^{\mathrm{cap}}$: Computational capability (in FLOPs per second) of edge server $e$.
    \item $M_{\mathrm{size}}$: Average memory access size (in bytes).
    \item $b_n$: Memory bandwidth (in bytes per second) of edge server $n$.
    \item $C^{\mathrm{ovh}}_e$: Constant overhead for operations, e.g., data transfer between CPU and GPU of edge server $e$, which cannot be overlapped with other operations (in seconds).
\end{itemize}
The transmission delay between LLM agent $n$ and mobile devices can be calculated as:
\begin{equation}
    \tau^{\mathrm{tran}}_{n}(\bx) = \frac{D(\bx)}{R_{n}},
\end{equation}
where $D(\bx)$ is the data size of the transmitted prompt $\bx$ and $R_{n}(t)$ is the achievable data rate between LLM agent $n$ and users. 

For an answer with $T$ tokens, the delay between mobile user $u$ and edge server with LLM model $n$ can be expressed as the followings, where for all $t \in [T]$,
\begin{equation}
\begin{aligned}
    \tau_{n}(\bx,t) =~&\tau^{\mathrm{comp}}_{n}(h(\bx,t)) + \tau_{u}^{\mathrm{tran}}(\hy_{t-1})+ \tau^{\mathrm{tran}}_{n}(p_{n,t} ),
\end{aligned}
\end{equation}
where we set $h(\bx, 0)=\bx$ and $\hat{y}_{0}=\bx$ for simplicity, and define $[T]\triangleq\{1,\cdots,T\}$.
% We  for notational completeness.

We use $\tau(\bx,\S)$ to denote the overall delay for token-level ensemble generation of prompt $x$ with LLM models subset $\S$, given by:
\begin{equation}
\tau(\bx,\S,t) \triangleq \max\limits_{n\in\S} \{ \tau_{n}(\bx,t) \}.
\end{equation}
In other words, the overall delay is determined by the the slowest edge LLM expert within the selected subset $\mathcal{S}$, as the framework cannot generate the next token until all experts have completed their computations.
The long-term system delay can be expressed as ($\tau^{\mathrm{gate}}$ is the delay of gating network computation):
\begin{equation}
    \tau(\bx,\S) = \frac{1}{T}\sum_{t=1}^T \tau(\bx,\S, t)+\tau^{\mathrm{gate}}.
\end{equation}
% Finally, 

Finally, each LLM application $m\in\mathcal{M}$ has its own latency requirement, given by the following constraint:
\begin{equation}
    \mathop{\mathbb{E}}\limits_{\bx\sim \mathcal{P}_{\bx,m}}[\tau(\bx,\S)]\leq\tau_{\max,m}, \quad \forall m\in\mathcal{M},
\end{equation}
% the system delay is constrained as follows:
% \begin{equation}
%     \mathop{\mathbb{E}}\limits_{\bx\sim \mathcal{P}_{\bx}}[\tau(\bx,\S)]\leq\tau_{\max},
% \end{equation}
where $\tau_{\max,m}$ is the delay limit (deadline) in response to arbitrary prompt $\bx$ for LLM application $m$.
\subsubsection{System Energy Consumption}

The energy consumption of our system mainly depends on the computational energy consumption of inference for edge LLM experts. The total energy consumption for token-level ensemble generation is:
\begin{equation}
E(\bx,\S,t) \triangleq \sum_{n\in\S} E_{n}^{\mathrm{comp}}(h(\bx,t)),
\end{equation}
where $E_{n}^{\mathrm{comp}}(h(\bx,t))$ represents the energy consumption to inference input $h(\bx,t)$ for the LLM agent $n$. The long-term system energy consumption can be represented as:
\begin{equation}
    E(\bx, \S)=\frac{1}{T}\sum_{t=1}^T E(\bx,\S,t),
\end{equation}
which is constrained as:
\begin{equation}
    \mathop{\mathbb{E}}\limits_{\bx\sim \mathcal{P}_{\bx}}[E(\bx, \S)]\leq E_{\max},
\end{equation}
where $E_{\max}$ is the energy budget.

\subsection{Problem Formulation}
Our objective is to jointly optimize the (coarse-grained level) expert selection $\mathcal{S}$ and train the gating parameters $\btheta$, considering the performance of LLM agents under the aforementioned system delay constraints and energy consumption constraint. 
% Specifically, we aim to select a subset of experts $\mathcal{S}$ and weight them by the gating network, which is parameterized by $\bm{\theta}$. Therefore, we should optimize both $\S$ and $\bm{\theta}$. 
To measure the quality of the ensemble answer generated by an expert subset $\S$.
For an answer with $T$ tokens generated from $\S$, the loss function of the ensemble prediction is defined as:
\begin{equation}
\begin{aligned}
\mathcal{L}(\boldsymbol{\theta},\mathcal{S},\bx, \by) \triangleq -\sum_{t=1}^{T} \log\left(\sum\limits_{n\in \mathcal{S}}  \omega_n(\bx, \bm{\theta}, \S) f_{n,y_t}(h(\bx,t))\right),
\end{aligned}
\end{equation}
where $f_{n,y_t}(h(\bx,t))$ denotes the probability of $y_t$ in $f_{n}(h(\bx,t))$. The expected loss function of the whole distribution can be denoted as:
\begin{equation}
\label{Eq-loss}
    \mathcal{L}_{\rm exp}(\boldsymbol{\theta},\mathcal{S}) =\mathop{\mathbb{E}}\limits_{(\bx, \by)\sim\mathcal{P}_{\bx,\by}} [\mathcal{L}(\boldsymbol{\theta},\mathcal{S},\bx, \by)].
\end{equation}
Since $\mathcal{P}_{\bx,\by}$ is not directly accessible in practice, we consider to use an empirical loss function instead:
\begin{equation}
\label{Emp-loss}
    \mathcal{L}(\boldsymbol{\theta},\mathcal{S}) = \frac{1}{|\mathcal{D}|}\sum_{(\bx, \by)\in\mathcal{D}} \mathcal{L}(\boldsymbol{\theta},\mathcal{S},\bx, \by),
\end{equation}
where $\mathcal{D}$ is a training dataset.
We now formulate the optimization problem:
\begin{tcolorbox}
    The \textit{joint gating parameter training and LLM expert selection problem} is formulated as:
\begin{equation}
    \begin{aligned}\label{problem}
    \textbf{OP:} \quad \min_{\boldsymbol{\theta}, \mathcal{S}}& \quad  
        \mathcal{L}(\boldsymbol{\theta},\mathcal{S}) \\
        \text{s.t.}&\quad  \mathop{\mathbb{E}}\limits_{\bx\sim \mathcal{P}_{\bx,m}}[\tau(\bx,\S)]\leq\tau_{\max,m},~\forall m\in\mathcal{M}, \\
        &\quad \mathop{\mathbb{E}}\limits_{\bx\sim \mathcal{P}_{\bx}}[E(\bx, \S)]\leq E_{\max}.
    \end{aligned}
\end{equation}
% where $\tau_{\max}$ is the delay limit and $E_{\max}$ is the energy consumption budget.
% The constraints ensure bounded computational delay and resource consumption.
\end{tcolorbox}

This problem is challenging due to the combinatorial nature of the expert subset selection, the non-convexity of the loss function, and the complex interplay between the system delay and energy consumption constraints. For a given input, the system must select the optimal subset of experts to minimize the loss function while satisfying the constraints on delay and energy consumption. Traditional optimization techniques may not be suitable for addressing these competing objectives simultaneously. Therefore, we propose a novel methodology to solve this problem effectively and efficiently.

 \begin{Remark}
We note that $\mathcal{P}_{\bx}$, which is required in the constraints of \eqref{problem}, is also not directly accessible in practice. However, based on our experimental results presented later, both $\tau(\bx,\S)$ and $E(\bx, \S)$ primarily depend on the length of the prompts, and their statistical significance becomes notable only when the prompt length exceeds $1024$  tokens. This implies that $\mathop{\mathbb{E}}\limits_{\bx\sim \mathcal{P}_{\bx,m}}[\tau(\bx,\S)]$ and $\mathop{\mathbb{E}}\limits_{\bx\sim \mathcal{P}_{\bx}}[E(\bx, \S)]$ are much easier to estimate in practice compared to $\mathcal{L}_{\exp}(\boldsymbol{\theta},\mathcal{S})$.
 \end{Remark}

\section{Methodology}
\label{Sec: Methodology}
MoE$^2$ consists of two key stages --- training and inference. In training stage, we solve Problem \eqref{problem} to attain the optimal overall response quality while subject to constraints. While in inference stage, we focus on balancing response quality with system costs.

\subsection{Optimal Solution Structures}
Define the optimal parameters $\bm\theta^*$ when fixing the subset $\S$ of edge LLMs: 
\begin{align}
        \bm{\theta}^{*}({\S}) \triangleq \arg\min\limits_{\bm{\theta}}  {\mathcal{L}}(\bm{\theta},\S),\quad \forall \S\subseteq\N.  \label{Eq: Optimality for subset}
\end{align}
We provide the theoretical analysis of the relation between $\bm{\theta}$ and $\mathcal{S}$. Then the following theorem extends 
\begin{Theorem}[Optimality for subset]
\label{Theorem: Optimality for subset}
Let the gating network $\boldsymbol{g}(\bx,\btheta)$ be an MLP with a sufficiently large width. For any given subset $\S\subseteq\N$ and the optimal parameters satisfy the following condition:
% \begin{align}
%     \bm{\theta}_{\S}^{*} = \arg\min\limits_{\bm{\theta}}  \hat{\mathcal{L}}(\bm{\theta},\S),
% \end{align}
% we have:
\begin{align}
   \frac{g_{n}(\bx, \bm{\theta}^{*}({\S}))}{\sum\limits_{n'\in \S}g_{n'}(\bx, \bm{\theta}^{*}({\S}))}=\frac{g_n(\bx, \bm{\theta}^{*}({\N}))}{\sum\limits_{n'\in \S}g_{n'}(\bx, \bm{\theta}^{*}({\N}))}, \forall n\in\S.
\end{align}
\end{Theorem}
The key idea of proving Theorem \ref{Theorem: Optimality for subset} is conducted by contradiction, constructing a parameter $\bm{\theta}'$ through universal approximation that satisfies the required conditions, ultimately leading to a contradiction.
For a detailed proof, please refer to Appendix \ref{AppenProof}. 
\begin{Remark}
Theorem \ref{Theorem: Optimality for subset} indicates that the optimal gating values $\bm{\theta}^*(\S)$ for any subset $\S$ can be derived directly from the optimal gating values of the complete set $\N$. Consequently, once the optimal parameter $\bm{\theta}^*(\N)$ for the entire set $\N$ of edge LLMs is obtained, the optimal gating values for any subset $\S$ can be determined without the need for additional training.
\end{Remark}
 Problem \eqref{problem} can be reduced to a combinatorial optimization over experts as:
\begin{equation} 
\label{Problem: Reduced Combinatorial Optimization Problem}
    \begin{aligned}
        \min_{\S}& \quad  \mathcal{L}(\btheta^*({\S}),\S)\\
        \text{s.t.}&\quad  \mathop{\mathbb{E}}\limits_{\bx\sim \mathcal{P}_{\bx,m}}[\tau(\bx,\S)]\leq\tau_{\max,m}, \quad \forall m\in\mathcal{M}, \\
        &\quad \mathop{\mathbb{E}}\limits_{\bx\sim \mathcal{P}_{\bx}}[E(\bx, \S)]\leq E_{\max},
    \end{aligned}
\end{equation}
Note that when searching for the optimal $\S^*$, Theorem \ref{Theorem: Optimality for subset} suggests that it is not necessary to repeatedly train the gating networks to obtain $\btheta^*({\S})$ for each $\S$. This significantly reduces the complexity of the original problem.

Now we aim to optimize subsets $\S$ to satisfy system constraints. We first introduce the property of monotonicity improvement of subsets as follows:
\begin{Theorem}[Monotonic Improvement]
    \label{Theorem: Monotonic Improvement} 
    Let the gating network $\boldsymbol{g}(\bx,\btheta)$ be an MLP with a sufficiently large width.
    For any pair of subsets $\S,\S'\subseteq\N$ such that $\S'\subseteq\S$, we have 
    \begin{align}
      \min_{\boldsymbol{\theta}} {\mathcal{L}}(\boldsymbol{\theta},\S)\leq \min_{\boldsymbol{\theta}} {\mathcal{L}}(\boldsymbol{\theta},\S').
    \end{align}
\end{Theorem}
\begin{proof}
Consider the following function $\boldsymbol{f}: \mathcal{X}\rightarrow \mathbb{R}^{|\mathcal{N}|}$, given by
\begin{equation}
\begin{aligned}
    {f}_n(\bx) =
    \begin{cases}
        \frac{g_n(\bx,{\btheta}^*(\S'))}{\sum\limits_{n'\in \S'} g_{n'}(\bx, \bm{\theta}^*({\S'}))}, & \text{if } n \in \S', \\
        0, & \text{if } n \in \S \setminus \S',
    \end{cases}
\end{aligned}
\end{equation}
where $\boldsymbol{h}(\bx)=[{h}_n(\bx)]_{n\in\mathcal{N}}$. Since $\boldsymbol{h}(\bx)$ has at most countably many discontinuities, it is Borel measurable
\cite{good1958mathematical}. According to the universal approximation theorem \cite{hornik1989multilayer}, there exists $\hat{\bm{\theta}}$ such that the MLP $\bg(\bx, \hat{\bm{\theta}})$ with a sufficiently large width can approximate $\boldsymbol{h}(\bx)$ to any degree of accuracy.
% \begin{align}
    
% \end{align}

Then, we have:
\begin{equation}
\begin{aligned}
    {\mathcal{L}}(\btheta^*({\S}),\S)
    \leq {\mathcal{L}}(\hat{\bm{\theta}},\S)
    ={\mathcal{L}}(\btheta^*(\S'),\S').
\end{aligned}
\end{equation}
% where $g_n(\bx, \bm{\theta}^*({\S}))$ is the $n$-th value of $\bg_{\S}(\bx,\hat{\bm{\theta}})$.
\end{proof}

\subsection{Algorithm Design}
% \textcolor{red}{I suggest to write three algorithms. One for training $\theta^*$; one for monotonic optimization; and then one for the inference stage. This way, it would be more clear to see that these two processes are separable.}

We present the algorithm for both training and inference stages. 
\begin{itemize}
    \item The training stage focuses on optimizing the gating network parameters, $\bm{\theta}$, and subset selection to enhance system performance while adhering to system constraints.
As suggested by Theorems \ref{Theorem: Optimality for subset} and \ref{Theorem: Monotonic Improvement}, the training process for $\bm{\theta}$ can be decoupled from the expert selection process. For the coarse-grained expert selection problem in \eqref{problem}, we propose utilizing the discrete monotonic optimization algorithm \cite{minoux2002discrete}.
\item In the inference stage, the goal is to perform fine-grained subset selection for each query, $\bx$, to fully exploit the heterogeneous capabilities of edge LLM experts. The structure of the proposed algorithm is illustrated in Fig. \ref{fig:framework}.
\end{itemize}

\subsubsection{Training Stage}
We first describe the \textit{training stage} as follows. As shown in Fig. \ref{fig:framework}, the training stage consists of two sequential steps. First, we utilize Algorithm \ref{Algorithm: Gating Network Optimization} to obtain parameter $\btheta^*(\mathcal{N})$ for all LLM agents $\N$, based on  the stochastic gradient descent method. 
\begin{algorithm}[t] 
\caption{Gating Network Training}
\label{Algorithm: Gating Network Optimization}
\begin{algorithmic}[1]
\STATE \textbf{Input:} Dataset $\mathcal{D}$ and Batch size $B$\;
\STATE \textbf{Output:} Optimized parameter $\bm{\theta}^*$.
% Parameter $\bm{\theta}$
\WHILE{not converged}
\STATE Sample a batch $\{\bx_i,\by_i\}_{i\in B}$ from $\mathcal{D}$.
\FOR{$i= 1, \dots, B$}
    \STATE Init $t=0$ and $h(\bx_i,t)=\{\bx_i\}$.
    \WHILE{true}
        \STATE Generate gating values $\bg(\bx_i, \bm{\theta})$ using the gating network.
        \STATE Compute normalized weights for all agents $\bm{\omega}(\bx_i, \bm{\theta}, \N) \gets \left[{g_n(\bx_i,\boldsymbol{\theta})}\Big/{\sum\limits_{n'\in \mathcal{\N}}g_{n'}(\bx_i,\boldsymbol{\theta})}\right]_{n\in\N}$.
        \STATE Send $\bx_i$ to all LLM experts and receive $\{f_n(h(\bx_i, t)\}_{n\in\N}$.
        \STATE Fuse the results: \\
        $F(\bx_i,\bm{\theta},\N,t) \gets \sum\limits_{n\in \mathcal{S}} \omega_n(\bx_i, \bm{\theta},\N) f_n(h(\bx_i,t))$.
        \STATE Sample next token $\hat{y}_{i,t+1}\sim F(\bx_i,\bm{\theta},\N,t)$.
        \STATE Add $\hat{y}_{i,t+1}$ to history: \\
        $h(\bx,t+1)\gets\{h(\bx,t),\hat{y}_{i,t+1}\}$,
        \STATE Accumulate loss: $\mathcal{L}(\boldsymbol{\theta},\mathcal{N},\bx_i, \by_i) \gets - \log\left(\sum\limits_{n\in \mathcal{S}}  \omega_n(\bx_i, \bm{\theta}, \S) f_{n,y_{i,t}}(h(\bx_i,t))\right)$.
        \IF{$\hat{y}_{i,t+1}$ is a stop token}
            \STATE \textbf{break}
        \ENDIF
    \ENDWHILE
    \STATE Accumulate batch loss $\mathcal{L}(\boldsymbol{\theta},\N) \gets \mathcal{L}(\boldsymbol{\theta},\N)+\frac{1}{|\mathcal{D}|}\mathcal{L}(\boldsymbol{\theta},\mathcal{N},\bx_i, \by_i)$.
\ENDFOR
\STATE Update $\bm{\theta}$: $\bm{\theta}\gets\bm{\theta}+\alpha\partial\mathcal{L}(\boldsymbol{\theta},\N)/\partial\bm{\theta}$.
\ENDWHILE
\end{algorithmic}
\end{algorithm}

% \textit{Training $\btheta^*(\mathcal{N})$}.

% Then, we optimize the subsets to be selected with optimized gating values to satisfy system constraints. From Theorem \ref{Theorem: Monotonic Improvement}, we know that the optimization of the subset selection has the property of monotonic improvement. Therefore, we can employ discrete monotonic optimization to address Problem \eqref{Problem: Reduced Combinatorial Optimization Problem}. 

% \textit{Discrete Monotonic Optimization for LLM Expert Selection}.
From Theorem \ref{Theorem: Monotonic Improvement}, we employ discrete monotonic optimization to address Problem \eqref{Problem: Reduced Combinatorial Optimization Problem}. 
 We define $G_1$ as the normal hull of the set defined by:
\begin{align}
G_1 \triangleq  \left\{ \mathcal{S}:   \mathop{\mathbb{E}}\limits_{\bx\sim \mathcal{P}_{\bx,m}}[\tau(\bx,\S)]\leq\tau_{\max,m}, ~\forall m\in\mathcal{M} \right\},
\end{align}
and $G_2$ as the normal hull of the set defined by:
\begin{align}
 G_2 \triangleq  \left\{\mathcal{S}: \mathop{\mathbb{E}}\limits_{\bx\sim \mathcal{P}_{\bx}}[E(\mathbf{x},\mathcal{S})]\leq E_{\max}\right\}.
\end{align}
Taking $G$ as their union, i.e., $G=G_1\cup G_2$, we then define
\begin{align}
\pi_G(\mathcal{S}) \triangleq \lambda \mathcal{S}, \text{ where } \lambda = \max\{\alpha > 0 \mid \alpha \mathcal{S} \in G\}.
\end{align}

Now we propose \textit{Subset Monotonic Optimization (SMO)}, as presented in Algorithm \ref{Algorithm: SMO}.
\begin{algorithm}[t] 
\caption{Subset Monotonic Optimization (SMO)}
\label{Algorithm: SMO}
\begin{algorithmic}[1]
\STATE \textbf{Initialization:} Select tolerance $\varepsilon \geq 0$ and $k \gets 1$ Let $CBV = -\infty$ and $T_1=\{\N\}$.\; 
\WHILE{True}
    \STATE \textbf{Step 1:}
    \STATE $\tilde{T}_k$ $\gets$ Remove all $\S$ from $T_k$ such that $ \mathcal{L}(\boldsymbol{\theta},\mathcal{S}) \geq CBV - \varepsilon$.

    \STATE \textbf{Step 2:}
    \IF{$\tilde{T}_k = \emptyset$}
        \IF{$CBV = -\infty$}
            \RETURN Problem \eqref{Problem: Reduced Combinatorial Optimization Problem} is infeasible
        \ELSE
            \RETURN Current best feasible solution $\bar{\S}_k$ as an $\varepsilon$-optimal solution.
        \ENDIF
    \ENDIF

    \STATE \textbf{Step 3:}
    \STATE Select $\S_k \in \arg\max\{ \mathcal{L}(\boldsymbol{\theta},\mathcal{S}) | \S \in \tilde{T}_k\}$
    \STATE Compute $\S'_k = \pi_G(\S_k)$ 
    \IF{$\S'_k = \S_k$}
        \RETURN $\S_k$ as an optimal solution.
    \ELSE
        \STATE Update new best feasible solution as $\S'_{k+1}$ and $CBV= \mathcal{L}(\boldsymbol{\theta},\S'_{k+1})$. 
    \ENDIF
    
    \STATE \textbf{Step 4:}
    \STATE $V_{k+1} \gets (\tilde{T}_k \setminus \{\S_k\}) \cup \{\S_k - (\S_{k}^i - \S_{k}^{'i})e^i, i = 1,\ldots,n\}$
    \FOR{$\S \in T_k \setminus\{\S_k\}$}
        \FOR{$i = 1$ to $n$}
            \IF{$\S\geq\S'_k$ while $\S_i <\S_k$}
                \STATE Remove $\S_k^i$ from $V_{k+1}$.
            \ENDIF
        \ENDFOR
    \ENDFOR
    \STATE $T_{k+1} \gets V_{k+1}$

    \STATE \textbf{Step 5:}
    \STATE Set $k \gets k + 1$.
\ENDWHILE
\end{algorithmic}
\end{algorithm}
Let $CBV$ denote the current best value, $e^i$ denote the $i$-th unit vector of the considered space, and $T_k$ be the proper vertex set of $\N$. First, we set the tolerance $\varepsilon$ and initialize the iteration counter $k$. The initial vertex set $T_1$ is initialized as the complete set of experts $\N$. In each iteration, we remove the subsets with loss values greater than the current best value $CBV - \varepsilon$ from the vertex set $T_k$. We then compute the remaining set $\tilde{T}_k$. If $\tilde{T}_k$ is empty, we check if the problem is infeasible and return the current best feasible solution $\bar{\S}_k$. Otherwise, we select the subset $\S_k$ with the maximum loss value from the remaining set $\tilde{T}_k$. We then compute the projection $\S'_k$ of $\S_k$ onto the feasible region $G$. If $\S'_k$ is equal to $\S_k$, we return $\S_k$ as the optimal solution. Otherwise, we update the current best feasible solution as $\S'_{k+1}$ and $CBV$ as the loss value of $\S'_{k+1}$. We then update the vertex set $T_{k+1}$ by removing $\S_k$ and adding the new subsets generated by the vertex $\S_k$.

\subsubsection{Inference Stage}
Now we describe the \textit{inference stage}. In this stage, we aim to compute the optimal subset selection for each query and generate the answer, which is represented in Algorithm \ref{Algorithm: Inference Stage}. 
\begin{algorithm}[t] 
\caption{Inference Stage}
\label{Algorithm: Inference Stage}
\begin{algorithmic}[1]
\STATE \textbf{Input:} Prompt $\bx$ and Optimal subset $\S$ and gating parameters $\bm{\theta}^*$.
\STATE \textbf{Output:} Answer $\hat{\bm{y}}$
\STATE Calculate gating values $g_{s_{n}}(\bx,\boldsymbol{\theta}^*)$ for $\S=\{s_1, s_2, \dots,s_n, \dots\}$.
\STATE Sort $\S$ according to the gating values: $\hat{\S}(\bx,\bm{\theta}^*)=\{s_n|g_{s_{n-1}}(\bx,\boldsymbol{\theta}^*) \geq g_{s_{n}}(\bx,\boldsymbol{\theta}^*)\}$.
\STATE Select Top-$k$ experts $\Gamma(\bx, \bm{\theta}^*, \S, k) \gets \{ \hat{s}_1, \hat{s}_2, \dots, \hat{s}_k \}$ where $\hat{s}_k\in \hat{\S}(\bx,\bm{\theta}^*)$.
\STATE Compute weights for the selected experts: $\bomega_{n}(\bx, \bm{\theta}^*, \S, k) \gets {g_n(\bx,\boldsymbol{\theta}^*)}\big/{\sum\limits_{n'\in \Gamma(\bx, \bm{\theta}^*, \S, k)}g_{n'}(\bx,\boldsymbol{\theta}^*)}$.
\STATE Init $t\gets0$.
\WHILE{true}
\STATE Fuse results: $F(\bx, \boldsymbol{\theta}^*, \mathcal{S}, k, t) \gets $ \\
\quad $\sum\limits_{n\in \Gamma(\bx, \bm{\theta}^*, \S, k)} \bm{\omega}_n(\bx, \bm{\theta}^*,\S, k) f_n(h(\bx,t))$.
\STATE Sample $\hy_{t+1}\sim F(\bx, \boldsymbol{\theta}^*, \mathcal{S}, k, t)$.
\STATE Collect answer $\hat{\bm{y}} \gets \hy_{t}$.
\IF{$\hy_{t}$ is a stop token}
\STATE \textbf{return}
\ENDIF
\ENDWHILE
\end{algorithmic}
\end{algorithm}
From Theorem \ref{Theorem: Optimality for subset}, we know that the optimal gating values for the subset can be directly obtained from the global optimal gating values, which corresponds to the first step of the inference stage. 

\textbf{Top-$k$ Mechanism.}
To further reduce system costs while preserving response quality, we employ the Top-$k$ mechanism \cite{DBLP:conf/iclr/ShazeerMMDLHD17} to select a subset of $k$ LLM experts with the highest gating values from subset $\S$. Specifically, given optimized parameter $\bm{\theta}^*$ and prompt $\bx$, we first sort all elements in subset $\S$ in ascending order with respect to gating values:
\begin{equation}
\begin{aligned}
    \hat{\S}(\bx,\bm{\theta}^*)=&\{s_n \mid g_{s_{n-1}}(\bx,\boldsymbol{\theta}^*) \geq g_{s_{n}}(\bx,\boldsymbol{\theta}^*)\}, \\
    &n\in\{1, 2, \dots,|\S|\},
\end{aligned}
\end{equation}
where $s_n\in\S$ is the $n$-th element of $\S$. Then the selected Top-$k$ experts can be denoted as: 
\begin{equation}
\Gamma(\bx, \bm{\theta}^*, \S, k) = \{ \hat{s}_1, \hat{s}_2, \dots, \hat{s}_k \}, s_k\in \hat{\S}(\bx,\bm{\theta}^*).
\end{equation}
Then we can compute the optimal weights for these $k$ LLM experts. For each expert $n\in \Gamma(\bx, \bm{\theta}^*, \S, k)$ , we have:
\begin{align}
    \bomega_{n}(\bx, \bm{\theta}^*, \S, k) = \frac{g_n(\bx,\boldsymbol{\theta}^*)}{\sum\limits_{n'\in \Gamma(\bx, \bm{\theta}^*, \S, k)}g_{n'}(\bx,\boldsymbol{\theta}^*)}.
\end{align}
To generate an answer, we fuse the outputs of the selected experts by a weighted combination as:
\begin{equation}
    F(\bx, \boldsymbol{\theta}^*, \mathcal{S}, k, t) = \sum\limits_{n\in \Gamma(\bx, \bm{\theta}^*, \S, k)} \omega_n(\bx, \bm{\theta}^*,\S, k) f_n(h(\bx,t)).
\end{equation}
Then we sample the next token $\hy_{t+1}$ from $F(\bx, \boldsymbol{\theta}^*, \mathcal{S}, k, t)$, i.e., $\hy_{t+1} \sim F(\bx, \boldsymbol{\theta}^*, \mathcal{S}, k, t)$. The process repeats until a stop token is generated.

\section{Simulation Results}
\label{Sec: Simulation Results}
In this section, we deploy MoE$^2$ on local devices and test the model performance. We use a novel cluster-based method to create domain experts, which exploits the strength of the MoE mechanism.
\subsection{Cluster-based Domain Experts}
To utilize the strengths of the MoE mechanism, experts are expected to be diverse and specialized in different tasks. However, it is challenging to acquire domain-specific experts due to the following reasons:
\begin{itemize}
    \item \textbf{Comprehensive LLM training:} Most existing LLM models are trained over diverse datasets, making it difficult to identify domain-specific experts. Although some models are fine-tuned on math or code-related tasks, they may not be sufficient for all tasks.
    \item \textbf{Undistinguishable domains:} The domain of a given prompt is not always clear, and the prompt may contain multiple domains. For example, should we consider a math question and a physics question as two separate domains or as a single domain?
    \item \textbf{Limited domain-specific datasets:} It is challenging to collect sufficient domain-specific data to train an expert even if the domain is distinguishable for prompts. To train a domain-specific expert, we need to label a large number of samples, which is time-consuming and expensive.
\end{itemize}
Therefore, we propose a cluster-based method to create domain-specific datasets from any existing datasets without manual labeling. Our main idea is to cluster the prompts based on their embedding similarity and fine-tune the LLM model on each cluster to create domain-specific experts. Specifically, we use the pre-trained mGTE model \cite{zhang2024mgte} to encode the prompts into embeddings, and apply $K$-means algorithm to cluster prompts based on their embeddings. In our simulation, we set $K=8$ for 8 clusters. Finally, we fine-tune the LLM model on each cluster to create domain-specific experts. Other details can be found in Appendix \ref{Appendix: Env}.

\subsection{Environment Setup}
We evaluate the performance of MoE$^2$ in a simulated environment. We construct a system with $N=8$ edge servers, each hosting an LLM agent. We pick Qwen2.5-3B-Instruct \cite{qwen2025qwen25technicalreport}, Qwen2.5-7B-Instruct \cite{qwen2025qwen25technicalreport}, Llama-3.2-3B-Instruct \cite{dubey2024llama} and Mistral-7B-Instruct-v0.3 \cite{jiang2023mistral} as base models. These LLM agents are duplicated and fine-tuned on each cluster to create domain-specific experts. The gating network is implemented as an MLP with two hidden layers. 

We choose MMLU \cite{hendrycks2020measuring} as the evaluation dataset, which is not involved in the fine-tuning process of the LLMs. More details on the fine-tuning process can be found in Appendix \ref{Appendix: Env}.
We retain $80\%$ data in MMLU as the training set for training the gating network, while the other $20\%$ are used as the testing set for the evaluating MoE$^2$'s performance.
Other detailed settings are listed in the Appendix \ref{Appendix: Env}.

\subsection{Simulation Results}
Here, we present the numerical results of MoE$^2$. We first evaluate the responsive accuracy of our model compared to the following baselines:
\begin{itemize}
    \item \textbf{Single Agent (S. A.):} This method evaluates each original model.
    \item \textbf{Majority Voting (M. V.):} This method assigns equal weights to all fine-tuned LLM agents and fuses the outputs by majority voting. The gating network is not used in this method.
    \item \textbf{Average Expert Accuracy (A. E. A.):} This metric calculates the average accuracy of each fine-tuned LLM agent on the test set. It reflects the average performance of experts.
\end{itemize}

Each model's responsive accuarcy on the MMLU dataset are shown in Fig. \ref{Fig: Accuracy on MMLU dataset}.

\begin{figure}
    \centering
    \includegraphics[width=1\linewidth]{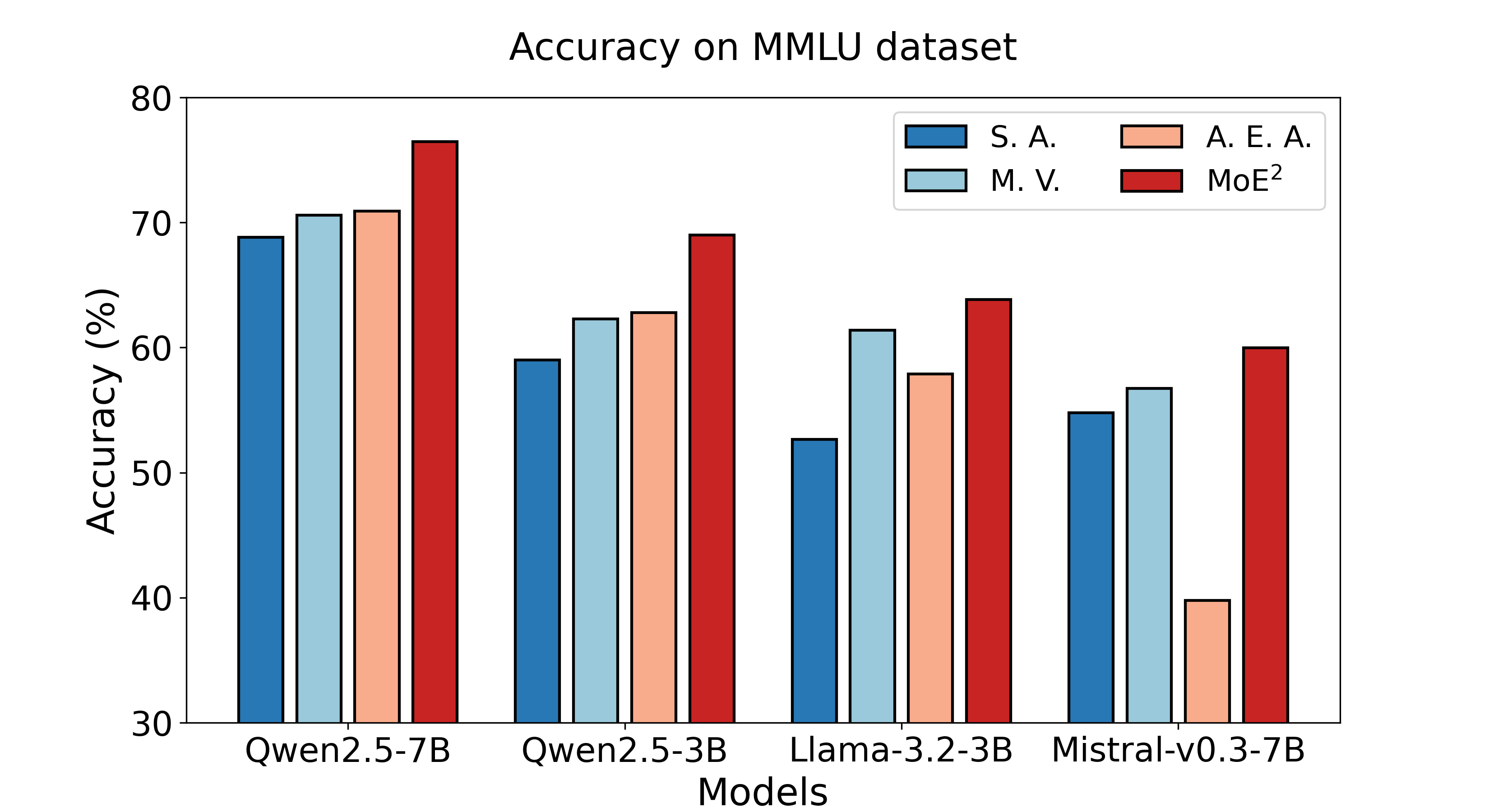}
    \caption{Performance comparison on the MMLU dataset. The MoE$^2$ architecture achieves the highest accuracy, which outperforms the Single Agent, Majority Voting and A. E. A. by $14.4\%$, $7.3\%$, $16.4\%$, respectively.}
    \label{Fig: Accuracy on MMLU dataset}
\end{figure}

We observe that the MoE$^2$ architecture achieves the highest accuracy, outperforming the Single Agent, Majority Voting, and A.E.A. by 14.4\%, 7.3\%, and 16.4\%, respectively. This validates that MoE$^2$ significantly enhances the performance of various LLM models.

\begin{figure}
    \centering
    \includegraphics[width=0.8\linewidth]{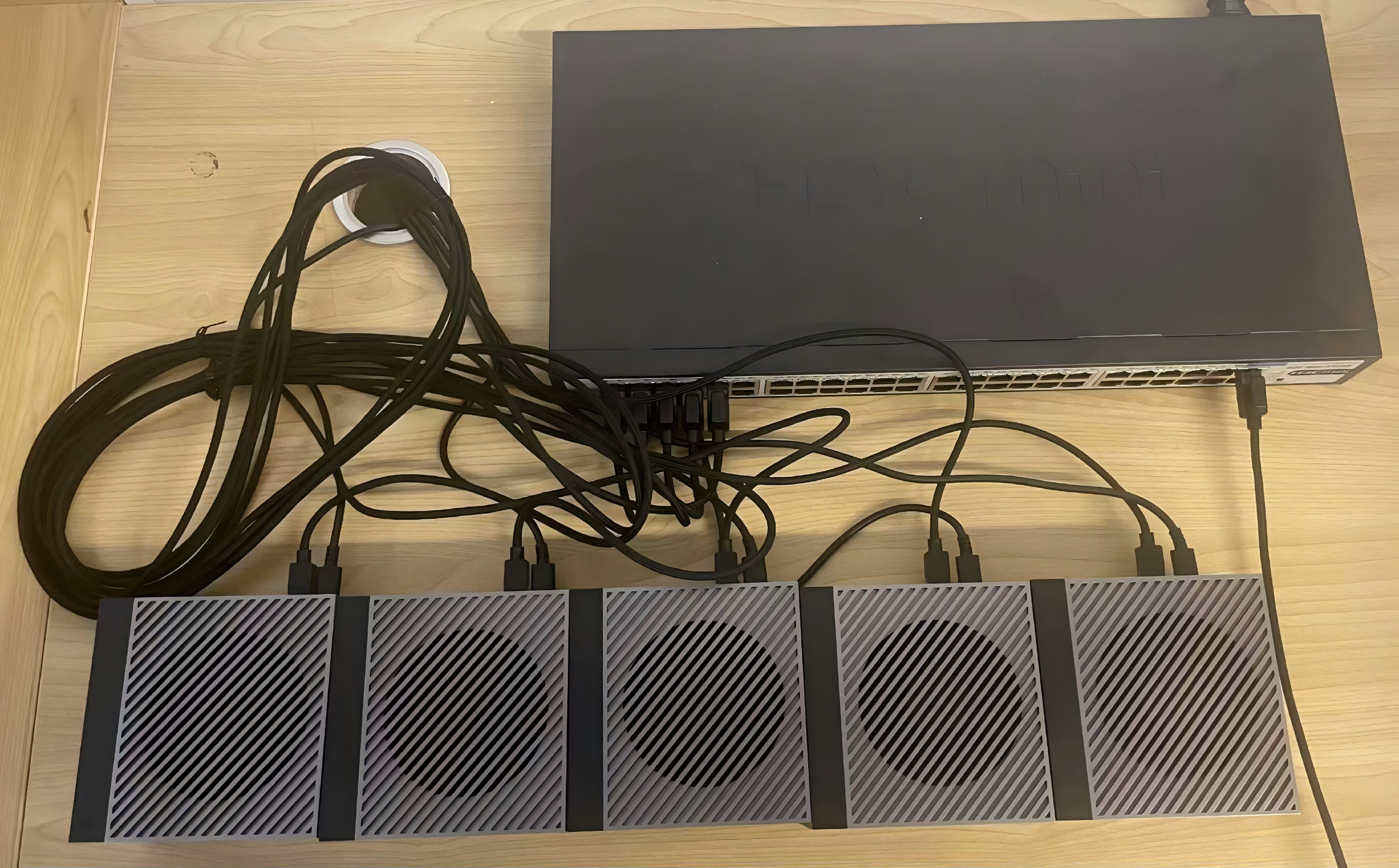}
    \caption{Implementation of an edge LLM testbed on NVIDIA Jetson AGX Orin 64GB, connected to NVIDIA GeForce RTX 4090 via a local area network.}
    \label{Fig: deploy_orin}
\end{figure}

\section{Experiment Results}\label{Sec:Experiment}
In this section, we implement our MoE$^2$ model on developed testbeds as shown in \ref{Fig: deploy_orin}. We evaluate its performance under various system constraints.
\begin{figure*}[ht] \label{Fig: delay_prompt_length}
    \centering
    \subfigure[Edge Delays under Different Prompt Lengths.]{\includegraphics[width=0.45\linewidth]{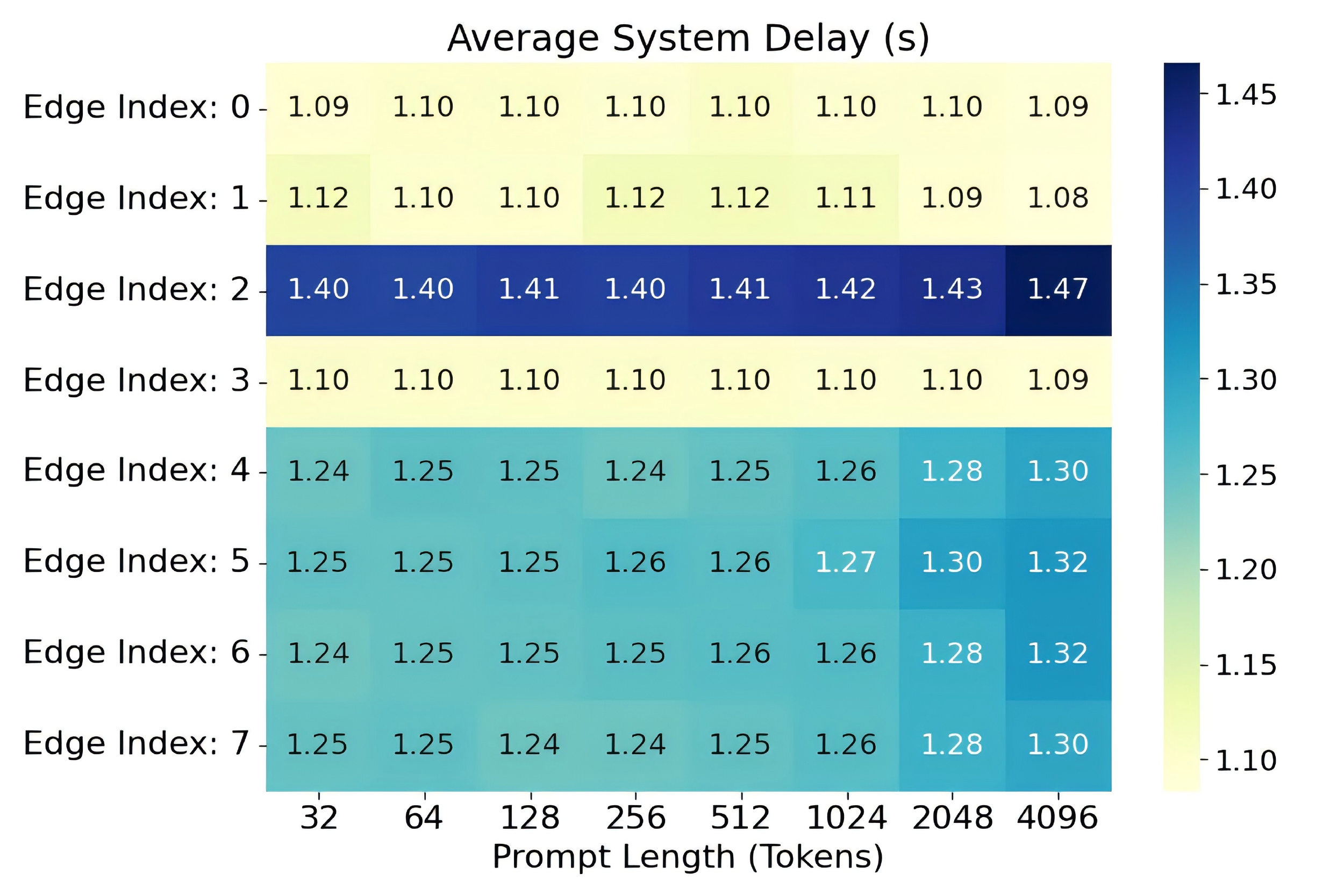}}
    \quad
    \subfigure[Edge Energy Consumptions under Different Prompt Lengths.]{\includegraphics[width=0.45\linewidth]{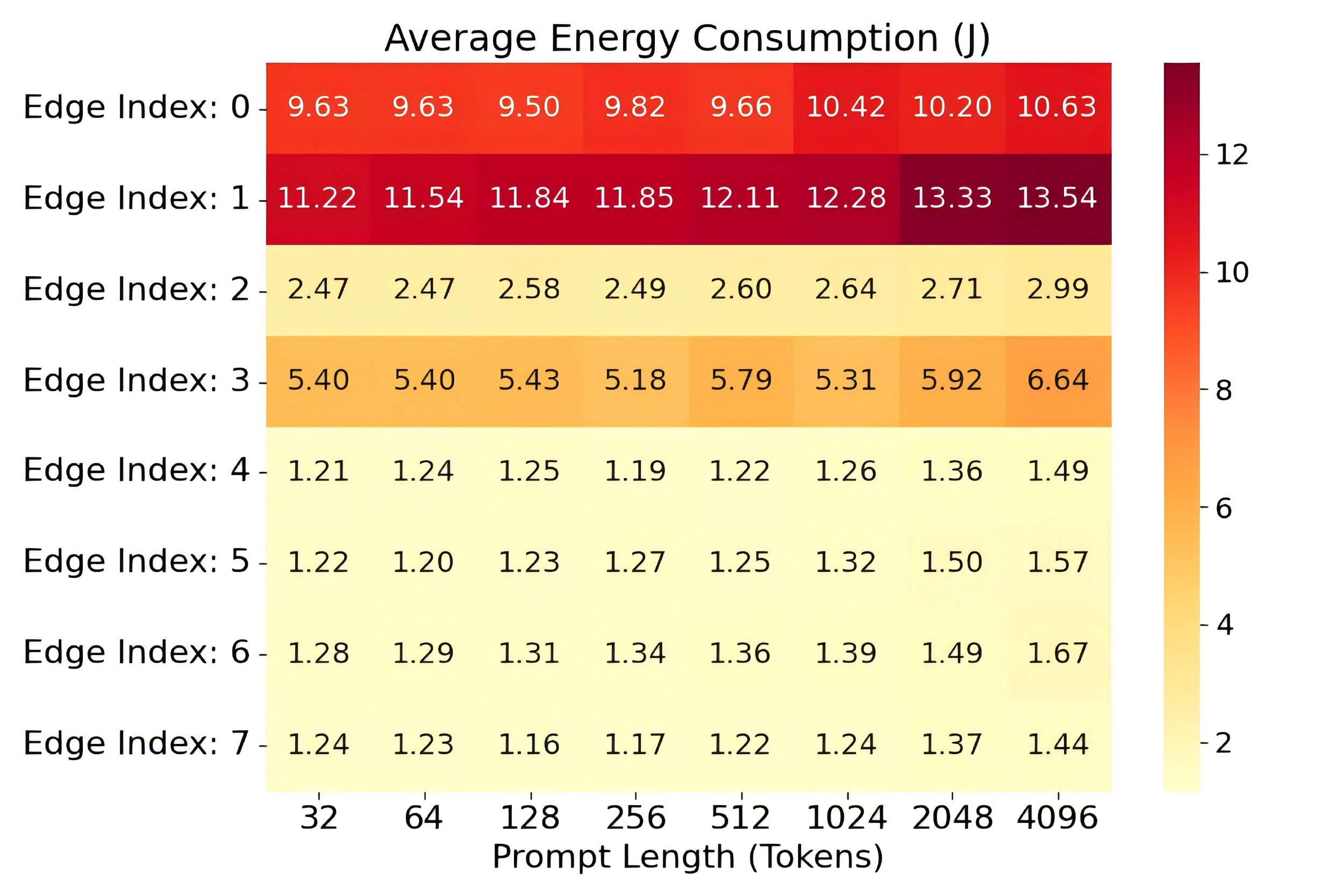}}
    \caption{Studies of edge delays and edge energy consumptions under different prompt lengths. Edges equipped with less powerful computational resources and larger-scale LLM models exhibit higher delays, while edges equipped with more powerful computational resources and larger-scale LLM models exhibit higher energy consumption.}
\end{figure*}

\subsection{Deployment of Edge Devices}

\begin{table}[t]
    \centering
    \caption{ Edge Servers Deployment}
    \begin{tabular}{ccc}
    \hline
        \textbf{Index}& \textbf{LLMs} &\textbf{Edge Devices}\\ 
        \hline
        0& Qwen2.5-7B-FT0-Q4 & NVIDIA GeForce RTX 4090\\
        1& Qwen2.5-7B-FT1-Q4 & NVIDIA GeForce RTX 4090\\ 
        2& Qwen2.5-7B-FT2-Q4 & NVIDIA Jetson AGX Orin 64GB\\
        3& Qwen2.5-3B-FT3-Q4 & NVIDIA GeForce RTX 4090\\
        4& Qwen2.5-3B-FT4-Q4 & NVIDIA Jetson AGX Orin 64GB\\
        5& Qwen2.5-3B-FT5-Q4 & NVIDIA Jetson AGX Orin 64GB\\
        6& Qwen2.5-3B-FT6-Q4 & NVIDIA Jetson AGX Orin 64GB\\
        7& Qwen2.5-3B-FT7-Q4 & NVIDIA Jetson AGX Orin 64GB\\
        \hline
    \end{tabular}
    \label{Table: deploy}
\end{table}
We implement MoE$^2$ on an edge LLM testbed with $N=8$ edge servers, each hosting an LLM agent. The edge servers are named according to the index of domain experts. The deployment details of these edge servers are summarized in Table \ref{Table: deploy}. Specifically, edge servers 0, 1, and 2 are deployed on high-performance platforms (NVIDIA GeForce RTX 4090), while the remaining edge servers (3, 4, 5, 6, 7) are deployed on NVIDIA Jetson AGX Orin 64GB devices, as shown in Fig. \ref{Fig: deploy_orin}.

We utilize two base LLM models in our deployment: Qwen2.5-7B \cite{qwen2025qwen25technicalreport} and Qwen2.5-3B \cite{qwen2025qwen25technicalreport}. These models are duplicated and fine-tuned individually before deployment on edge devices, as in Section \ref{Sec: Simulation Results}. To optimize memory usage and inference latency for edge deployment, we apply 4-bit quantization to the models. For example, ``Qwen2.5-7B-FT0-Q4'' indicates the Qwen2.5-7B model fine-tuned for cluster 0 and deployed with 4-bit quantization.

\subsection{Experiment Results}
We first analyze the edge delays and energy consumptions across different servers under varying prompt lengths.  As illustrated in Fig. \ref{Fig: delay_prompt_length}, edge servers 0 and 1 with more powerful platforms (NVIDIA RTX 4090) demonstrate significantly lower latency and higher energy consumption compared with edge server 2 with same scale LLM models. Edge server 4-7 equipped with  smaller-scale LLM models (3B) exhibits lower delay and energy consumption compared to edge servers 2 with the same platform. Furthermore, Fig. \ref{Fig: delay_prompt_length} reveals that the edge delays and energy consumptions primarily depend on the length of the prompts. In addition, their statistical significance becomes notable only when the prompt length exceeds $1024$ tokens.
This finding indicates minimal variation in processing delays between original prompts and historical context information.

Subsequently, we conduct LLM inference experiments in our collaborative inference framework for edge LLM models. Here we aim to test the model performance under the variation of delay and cost constraints. We evaluate our MoE$^2$ framework with the proposed SMO subset selection algorithm and Top-$k$ selection mechanism on various constraints of end-to-end delay $\tau_{\max}$ and energy consumption $E_{\max}$. The baselines are as follows:
\begin{itemize}
    \item \textbf{SMO with Majority Voting (SMO M. V. ): }  This method assigns equal weights to LLM models selected by the proposed SMO selection mechanism. 
    \item \textbf{Random Subset Selection with Majority Voting (Rand. M. V. ):} This method assigns equal weights to LLM models with randomly selected subsets. 
\end{itemize}

% Table for Max Delay = 1
\begin{table}[t]
\centering
\caption{Accuracy Comparison for $\tau_{\max}$ = 1 sec }
\begin{tabular}{cccccccc}
\hline
$\mathbf{E_{\max} (J)}$ & \textbf{5} & \textbf{10} & \textbf{15} & \textbf{20} & \textbf{25} & \textbf{35} & \textbf{50} \\
\hline
MoE$^2$& 58.3 & 59.3 & 59.3 & 59.3 & 59.3 & 59.9 & 59.9 \\

SMO M. V.& 58.3 & 59.1 & 59.1 & 59.1 & 59.1 & 59.4 & 59.4 \\

Rand. M. V. & 57.5 & 58.9 & 56.8 & 56.1 & 57.0 & 59.0 & 57.5 \\
\hline
\end{tabular}
\label{tab:accuracy_delay_1}
\end{table}

% Table for Max Delay = 2
\begin{table}[t]
\centering
\caption{Accuracy Comparison for $\tau_{\max}$ = 2 sec }
\begin{tabular}{cccccccc}
\hline
$\mathbf{E_{\max} (J)}$ & \textbf{5} & \textbf{10} & \textbf{15} & \textbf{20} & \textbf{25} & \textbf{35} & \textbf{50} \\
\hline
MoE$^2$& 58.3 & 59.4 & 59.4 & 66.1 & 71.6 & 71.7 & 71.8 \\

SMO M. V.& 58.3 & 59.0 & 59.2 & 64.7 & 71.6 & 69.5 & 69.5 \\

Rand. M. V. & 58.0 & 58.4 & 58.6 & 58.4 & 60.4 & 61.2 & 63.3 \\

\hline
\end{tabular}
\label{tab:accuracy_delay_2}
\end{table}

% Table for Max Delay = 3
\begin{table}[t]
\centering
\caption{Accuracy Comparison for $\tau_{\max}$ = 3 sec }
\begin{tabular}{cccccccc}
\hline
$\mathbf{E_{\max} (J)}$ & \textbf{5} & \textbf{10} & \textbf{15} & \textbf{20} & \textbf{25} & \textbf{35} & \textbf{50} \\
\hline
MoE$^2$& 58.3 & 66.2 & 66.2 & 66.2 & 71.6 & 71.7 & 71.8 \\

SMO M. V.& 58.3 & 66.2 & 66.2 & 66.2 & 71.6 & 69.5 & 69.5 \\

Rand. M. V.  & 58.0 & 61.5 & 61.0 & 60.2 & 60.7 & 62.6 & 62.9 \\

\hline
\end{tabular}
\label{tab:accuracy_delay_3}
\end{table}

The testbed experiment results are shown in Table \ref{tab:accuracy_delay_1}, \ref{tab:accuracy_delay_2}, and \ref{tab:accuracy_delay_3}. From Table \ref{tab:accuracy_delay_1}, we observe that both MoE$^2$ and SMO M. V.  consistently outperform  Rand. M.V. across all energy constraints under the tightest delay constraint of 1s, demonstrating the effect of the proposed SMO algorithm. In addition, MoE$^2$ demonstrates strong performance at higher energy constraints, achieving a 4.2\% improvement over Rand. M.V. at energy constraints of 50J per prompt. This indicates that MoE$^2$ can effectively utilize the increased energy budget to activate a more optimal set of experts.
With a more relaxed delay constraint of 2s in Table \ref{tab:accuracy_delay_2}, MoE$^2$ outperforms SMO M. V. by an average of 1.38\% and Rand. M. V. by an average of 9.4\%. The substantial improvements over Rand. M. V. highlights the effectiveness of the MoE$^2$ framework in utilizing the increased flexibility to select a more optimal set of experts with the proposed SMO algorithm.
With the most relaxed delay constraint shown in Table \ref{tab:accuracy_delay_3}, MoE$^2$ demonstrates its ability to effectively utilize the added time flexibility. While it performs on par with SMO M. V. in lower and mid-range energy constraints, it significantly outperforms Rand. M. V. up to 18.0\% at 25J. With the increased energy budget (35J and 50J), MoE$^2$ surpasses SMO M. V. by 3.2\% and 3.3\% respectively. This shows the effectiveness of the MoE$^2$ framework in utilizing the selected LLM subsets with gating networks at overall relaxed resource constraints.

In summary, the results across all three tables strongly support the effectiveness of the proposed MoE$^2$ framework. MoE$^2$ consistently outperforms both the SMO M. V. and Random M. V. methods, particularly as the energy constraint $E_{\max}$ increases and the delay constraint $\tau_{\max}$ is relaxed. This highlights the effectiveness of MoE$^2$ and the SMO selection mechanism in optimizing the utilization of multiple LLM models on edge devices. The significant improvement over Rand. M.V. underscores the importance of intelligent expert selection. While the improvement over M.V. might seem modest in some cases, it is important to remember that M.V. itself is a strong baseline equipped with proposed SMO subset selection. These findings firmly establish MoE$^2$ as a highly promising solution for deploying LLMs in resource-constrained edge computing environments.

\section{Conclusions} \label{Sec: Con}

In this paper, we propose the first MoE-aided collaborative inference framework for edge LLMs by optimally designing the gating network and the two-level expert selection mechanism.
Through theoretical analysis, we have proven that the optimality of gating parameters for the entire LLM sets can be extended to any subset, based on which we develop the MoE$^2$ framework with SMO.
Extensive experiments on real hardware platforms demonstrate that our proposed MoE$^2$ framework achieves superior performance compared to baseline methods under diverse system cost constraints.
Overall, by leveraging collaborative inference and optimized resource allocation, MoE$^2$ paves the way for efficient and effective small-scale LLMs applications in resource-constrained edge computing scenarios. 

As the first work in this domain, there are several promising directions for future research. First, it would be valuable to adapt MoE$^2$ to dynamic system conditions, such as fluctuating user requests and varying edge server loads. Second, further exploration into edge learning and cloud-edge deployment of MoE$^2$ could enhance its practical applicability. Third, the deployment of multi-modal LLMs on edge servers presents an exciting avenue for investigation, along with exploring the potential of MoE$^2$ to support LLM agents without requiring fine-tuning.

\appendices

\begin{figure*}[t]
    \centering
    \includegraphics[width=18cm]{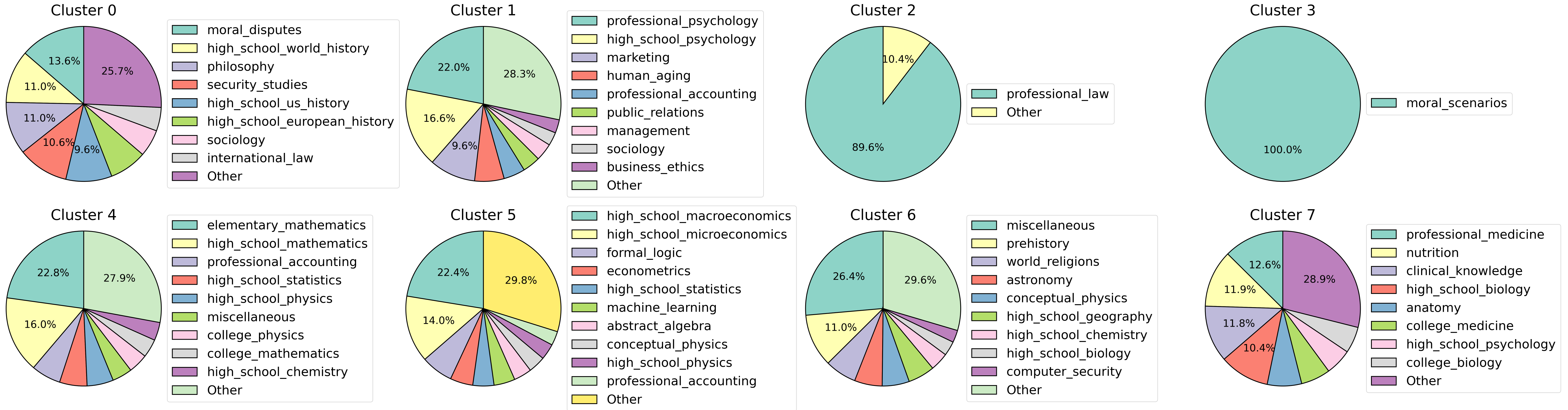}
    \caption{Clustering Result of MMLU. The prompts are clustered based on their embedding. The results show that similar subjects fall into one cluster. For example, cluster 4 mainly contains math-related and physical-related subjects and cluster 7 includes medical-related subjects. This result coincides with the intuition of us to distinguish domains.}
    \label{Fig: Clustering}
\end{figure*}

\section{Environment Setup}
\label{Appendix: Env}

\textit{Fine-tuning LLM Agents.} We fine-tune the LLM agents using eight NVIDIA Tesla A800 80G GPUs for each cluster to develop domain-specific experts. The steps are as follows:
\begin{enumerate}
    \item \textbf{Dataset Fusing}. We combine the datasets of different domains to create a comprehensive dataset, including MMLU-pro \cite{wang2024mmluprorobustchallengingmultitask}, ARC \cite{clark2018thinksolvedquestionanswering}, SciQ \cite{welbl2017crowdsourcingmultiplechoicescience} and AGIEval \cite{zhong2023agievalhumancentricbenchmarkevaluating}. The MMLU dataset is not included in this comprehensive dataset, as it is used for performance evaluation.
    \item \textbf{Embedding-based Clustering}. We use the pre-trained mGTE model \cite{zhang2024mgte} to encode the prompts into embeddings. Then, we apply the $K$-means algorithm to cluster the prompts based on their embeddings. We set the number of clusters to $8$, corresponding to the number of edge servers in the testbed described in Section \ref{Sec:Experiment}.
    \item \textbf{Model Fine-tuning}. We fine-tune the LLM model on each cluster to create domain-specific experts. We use LoRA \cite{hu2021loralowrankadaptationlarge} method to fine-tune the LLM model. The fine-tuning parameters are shown in Table \ref{Table: Fine-tuning Parameters}.
\end{enumerate}

\begin{table}[t]
    \centering
    \caption{Fine-tuning Parameters}
    \begin{tabular}{cc}
    \hline
        \textbf{Parameters} & \textbf{Values} \\ 
        \hline
        LoRA Rank & $256$ \\
        LoRA Alpha & $512$ \\ 
        LoRA Dropout & $0.05$ \\
        Cutoff Length & $8192$ \\
        Training Epochs & $3$ \\
        Batch Size & $8$ \\
        Gradient Accumulation Steps & $16$ \\
        Optimizer & AdamW \\
        Learning Rate Scheduler & CosineAnnealingLR \\
        Init Learning Rate & $5\times10^{-5}$ \\
        Max Gradient Norm & $1.0$ \\
        \hline
    \end{tabular}
    \label{Table: Fine-tuning Parameters}
\end{table}

After fine-tuning, we obtain $8$ domain-specific experts, each corresponding to a cluster. We use these experts as the LLM agents in the MoE system. To validate the effectiveness of this method, we calculate the embeddings of the MMLU dataset and present the clustering results in Fig. \ref{Fig: Clustering}. From the results, we observe that similar subjects fall into one cluster. For example, cluster 4 mainly contains math-related and physical-related subjects and cluster 7 includes medical-related subjects. This result coincides with the intuition of us to distinguish domains. Most importantly, the dataset used to fine-tune the LLM agents does not include the MMLU dataset, which exhibits strong generalization ability of this clustering method.

\textit{Training Gating Network.} For the gating network, we use a feed-forward neural network with two hidden layers. We also use a residual connection between the input of each layer and the input of the last layer. All residual activations are concatenated before the final output layer. The learning rate is scheduled by the ReduceLROnPlateau method, which reduces the learning rate by a factor if the validation loss does not improve for a certain number of steps. The parameters of the gating network are shown in Table \ref{Table: Gating Network Parameters}.

\textit{Collaborative Inference Architecture.}  We conduct collaborative inference with edge servers with Nvidia RTX 4090 GPUs and Nvidia Jetson AGX Orins. Detailed configuration is presented in Table \ref{Table: deploy}.

\begin{table}[t]
    \centering
    \caption{Gating Network Parameters}
    \begin{tabular}{cc}
    \hline
        \textbf{Parameters} & \textbf{Values} \\ 
        \hline
        Input (Embedding) Dimension & $1024$ \\
        Hidden Dimension & $128$ \\
        Output Dimension & $8$ \\
        Dropout & $0.1$ \\
        Activation Function & PReLU \\
        Optimizer & AdamW \\
        Learning Rate Scheduler & ReduceLROnPlateau \\
        Init Learning Rate & $10^{-5}$ \\
        Learning Rate Decay Factor & $0.8$ \\
        Learning Rate Decay Patience & $50$ \\
        Gradient Clip & $1.0$ \\
        Gradient Accumulation Steps & $8$ \\
        Batch Size & $8$ \\
        \hline
    \end{tabular}
    \label{Table: Gating Network Parameters}
\end{table}

\section{Proof of Theorem \ref{Theorem: Optimality for subset}}\label{AppenProof}
We prove this theorem by contradiction. Suppose that the above theorem does not hold, then there exists an optimal $\hat{\bm{\theta}}$ on $\S$ such that:
\begin{align}
    \mathcal{L}(\hat{\bm{\theta}},\S)<\mathcal{L}(\bm{\theta}^{*}(\mathcal{N}),\S).
\end{align}
where $\bm{\theta}^{*}(\mathcal{N})$ satisfies
\begin{align}
   \frac{g_{n}(\bx, \bm{\theta}^{*}({\S}))}{\sum\limits_{n'\in \S}g_{n'}(\bx, \bm{\theta}^{*}({\S}))}=\frac{g_n(\bx, \bm{\theta}^{*}({\N}))}{\sum\limits_{n'\in \S}g_{n'}(\bx, \bm{\theta}^{*}({\N}))}, \forall n\in\S.
\end{align}
This means that there exists a parameter $\hat{\bm{\theta}}$ that achieves a lower loss on the subset $\S$ compared to $\bm{\theta}^{*}(\N)$.
% where
%     \begin{align}
%         \min_{\S}& \quad  \mathcal{L}(\btheta^*({\S}),\S)\\
%         \text{s.t.}&\quad \mathop{\mathbb{E}}\limits_{\bx\sim \mathcal{P}_{\bx}}[\tau(\bx,\S)]\leq\tau_{\max}, \\
%         &\quad \mathop{\mathbb{E}}\limits_{\bx\sim \mathcal{P}_{\bx}}[E(\bx, \S)]\leq E_{\max},
%     \end{align}
    
 Let $c(\bx)={\sum\limits_{n\in\mathcal{S}}g_n(\bx, \bm{\theta}^{*}(\N))}\big/{\sum\limits_{n\in\mathcal{S}}g_n(\bx, \hat{\bm{\theta}})}$. 
 Consider the function $\boldsymbol{h}: \mathcal{X}\rightarrow \mathbb{R}^{|\mathcal{N}|}$, given by
 % By invoking the Universal Approximation Theorem \cite{hornik1989multilayer}, there exists parameters $\bm{\theta}'$ such that:
\begin{align}
{h}_n(\bx) = \begin{cases}
    c(\bx)g_n(\bx, \hat{\bm{\theta}}), & n\in\mathcal{S}, \\
    g_n(\bx, \bm{\theta}^{*}(\N)), & n\in\mathcal{N}\setminus\mathcal{S},
\end{cases}
\end{align}
Since $\boldsymbol{h}(\bx)$ has at most countably many discontinuities, it is Borel measurable
\cite{good1958mathematical}. Invoking the universal approximation theorem \cite{hornik1989multilayer}, there exists $\bm{\theta}'$ such that the MLP $\bg(\bx, \bm{\theta}')$ with a sufficiently large width can approximate $\boldsymbol{h}(\bx)$ to any degree of accuracy.

This
which ensures that
\begin{align}
    \sum\limits_{n\in\mathcal{S}}g_n(\bx, \bm{\theta}')=\sum\limits_{n\in\mathcal{S}}g_n(\bx, \bm{\theta}^{*}(\N)), \forall \bx\in\mathcal{X}.\label{Eq-5}
\end{align} Thus, for all $\bx\in\mathcal{X}$, we have
\begin{equation}
\begin{aligned}
    \frac{c(\bx)g_n(\bx, \hat{\bm{\theta}})}{\sum\limits_{n'\in \S}c(\bx)g_{n'}(\bx, \hat{\bm{\theta}})}=\frac{g_n(\bx, \hat{\bm{\theta}})}{\sum\limits_{n'\in \S}g_{n'}(\bx, \hat{\bm{\theta}})}.
\end{aligned}
\end{equation}
From the definition of $\mathcal{L}(\bm{\theta},\S)$, we have:
\begin{equation}
\begin{aligned}
    &  \mathcal{L}(\hat{\bm{\theta}},\S)=  \mathcal{L}(\bm{\theta}',\S)\\
    % &= \mathbb{E}\left[ \sum_{t=1}^{T} -\log\left(\sum\limits_{n\in \S}\frac{g_{n}(\bx, \bm{\theta}')}{\sum\limits_{n'\in \S}g_{n'}(\bx, \bm{\theta}')} f_{n,y_t}(h(\bx,t)) \right)\right] \\
    &= \frac{1}{|\mathcal{D}|}\sum_{(\bx, \by)\in\mathcal{D}}\left[ \sum_{t=1}^{T} -\left(\log\left(   \sum\limits_{n\in \S}g_{n}(\bx, \bm{\theta}') f_{n,y_t}(h(\bx,t))\right) - \right.\right. \\ 
    & \quad\quad\quad \left.\left. \log\left(\sum\limits_{n'\in \S}g_{n'}(\bx, \bm{\theta}')\right)\right)\right]\\
    &< \mathcal{L}(\bm{\theta}^{*}(\N), \S) \\
    &= \frac{1}{|\mathcal{D}|}\sum_{(\bx, \by)\in\mathcal{D}}\left[ \sum_{t=1}^{T} -\log\left(\sum\limits_{n\in \S}\frac{g_n(\bx, \bm{\theta}^{*}(\N))}{\sum\limits_{n'\in \S}g_{n'}(\bx, \bm{\theta}^{*}(\N))} f_{n,y_t}(h(\bx,t))\right)\right] \\
    &= \frac{1}{|\mathcal{D}|}\sum_{(\bx, \by)\in\mathcal{D}}\left[ \sum_{t=1}^{T} -\left(\log\left(   \sum\limits_{n\in \S}g_n(\bx, \bm{\theta}^{*}(\N)) f_{n,y_t}(h(\bx,t))\right)\right.\right. \\
    & \quad\quad\quad \left.\left. - \log\left(\sum\limits_{n'\in \S}g_{n'}(\bx, \bm{\theta}^{*}(\N))\right)\right)\right].
\end{aligned}
\end{equation}
From \eqref{Eq-5}, it follows
\begin{equation}
    \begin{aligned} \label{ineq1}
    & \frac{1}{|\mathcal{D}|}\sum_{(\bx, \by)\in\mathcal{D}}\left[ \sum_{t=1}^{T} -\log\left(\sum\limits_{n\in \S}g_n(\bx, \hat{\bm{\theta}}) f_{n,y_t}(h(\bx,t))\right) \right] \\
    &< \frac{1}{|\mathcal{D}|}\sum_{(\bx, \by)\in\mathcal{D}}\left[ \sum_{t=1}^{T} -\log\left(\sum\limits_{n\in \S}g_n(\bx, \bm{\theta}^{*}(\N)) f_{n,y_t}(h(\bx,t))\right) \right].\\
    \end{aligned}
\end{equation}
On the other hand, from the optimality of ${\bm{\theta}}^*(\mathcal{N})$, we have:
\begin{equation}
\begin{aligned}
    &\quad \mathcal{L}(\bm{\theta
    }^*(\mathcal{N}),\mathcal{\N}) \\
    % &= \frac{1}{|\mathcal{D}|}\sum_{(\bx, \by)\in\mathcal{D}} \left[ \sum_{t=1}^{T} -\log\left(\sum\limits_{n\in \N}\frac{g_n(\bx, \bm{\theta}^{*}(\N))}{\sum\limits_{n'\in \N}g_{n'}(\bx, \bm{\theta}^{*}(\N))} f_{n,y_t}(h(\bx,t))\right) \right]\\
    &= \frac{1}{|\mathcal{D}|}\sum_{(\bx, \by)\in\mathcal{D}} \left[ \sum_{t=1}^{T} -\log\left(  \sum\limits_{n\in \S}g_n(\bx, \bm{\theta}^{*}(\N)) f_{n,y_t}(h(\bx,t)) \right.\right. \\ 
    & \left.\left. \quad\quad\quad\quad + \sum\limits_{n\in \N \setminus \S}g_n(\bx, \bm{\theta}^{*}(\N)) f_{n,y_t}(h(\bx,t))\right)\right. \\
    & \left. \quad\quad\quad\quad + \log\left(\sum\limits_{n'\in \N}g_{n'}(\bx, \bm{\theta}^{*}(\N))\right)\right] \\
    & \leq \mathcal{L}(\bm{\theta}',\N) \\
    &= \frac{1}{|\mathcal{D}|}\sum_{(\bx, \by)\in\mathcal{D}} \left[ \sum_{t=1}^{T} -\log\left({\sum\limits_{n\in \N}g_n(\bx, \hat{\bm{\theta}})f_{n,y_t}(h(\bx,t))}\Big/ \right.\right.\\
    &\quad\quad \left.\left. \left({\sum\limits_{n'\in \S}g_{n'}(\bx, \hat{\bm{\theta}})+\sum\limits_{n'\in \N\setminus\S}g_{n'}(\bx, \bm{\theta}^{*}(\N))} \right)\right)\right] \\
    &= \frac{1}{|\mathcal{D}|}\sum_{(\bx, \by)\in\mathcal{D}} \left[ \sum_{t=1}^{T} - \log\left(   \sum\limits_{n\in \S}g_n(\bx, \hat{\bm{\theta}}) f_{n,y_t}(h(\bx,t)) \right.\right. \\
    & \quad\quad \left.\left. + \sum\limits_{n\in \N \setminus \S}g_n(\bx, \bm{\theta}^{*}(\N)) f_{n,y_t}(h(\bx,t))\right)\right.\\
    &\left.+ \log\left(\sum\limits_{n'\in \S}g_{n'}(\bx, \hat{\bm{\theta}})+\sum\limits_{n'\in \N\setminus\S}g_{n'}(\bx, \bm{\theta}^*(\N))\right)\right]
\end{aligned}
\end{equation}
It follows that
\begin{equation}
\begin{aligned} \label{ineq2}
    &\frac{1}{|\mathcal{D}|}\sum_{(\bx, \by)\in\mathcal{D}} \left[ \sum_{t=1}^{T} -\log\left(   \sum\limits_{n\in \S}g_n(\bx, \bm{\theta}^{*}(\N)) f_{n,y_t}(h(\bx,t)) \right) \right.\\
    & \quad\quad + \left. \log\left(\sum\limits_{n'\in \N}g_{n'}(\bx, \bm{\theta}^{*}(\N))\right) \right]\\
    &\leq \frac{1}{|\mathcal{D}|}\sum_{(\bx, \by)\in\mathcal{D}} \left[ \sum_{t=1}^{T} -\log\left(\sum\limits_{n\in \S}g_n(\bx, \hat{\bm{\theta}}) f_{n,y_t}(h(\bx,t)) \right)\right. \\ 
    &+ \left. \log\left(\sum\limits_{n'\in \S}g_{n'}(\bx, \hat{\bm{\theta}})+\sum\limits_{n'\in \N\setminus\S}g_{n'}(\bx, \bm{\theta}^{*}(\N))\right)\right]
\end{aligned}
\end{equation}
Again from \eqref{Eq-5}, we observe that \eqref{ineq1} contradicts with \eqref{ineq2}, which completes the proof.

\bibliographystyle{IEEEtran}
\bibliography{ref}

% \clearpage

\end{document}

\maketitle
\section{Appendix A}

\begin{Theorem}[Optimality for subset]
For any given subset $\S\subseteq\N$ and the following optimal parameters:
\begin{align}
    \label{Eq: Optimality for subset}
    \bm{\theta}^{*}(\N) = \arg\min\limits_{\bm{\theta}}  \mathcal{L}(\bm{\theta},\N). \\
    \bm{\theta}^{*}(\S) = \arg\min\limits_{\bm{\theta}}  \mathcal{L}(\bm{\theta},\S),
\end{align}
we have:
\begin{align}
    \frac{g_n(\bx, \bm{\theta}^{*}(\N))}{\sum\limits_{n'\in \S'}g_{n'}(\bx, \bm{\theta}^{*}(\N))}=\frac{g_{n}(\bx, \bm{\theta}^{*}(\S))}{\sum\limits_{n'\in \S'}g_{n'}(\bx, \bm{\theta}^{*}(\S))}, \forall n\in\S.
\end{align}
\end{Theorem}
\begin{proof}
We prove this theorem by contradiction. Assume the above theorem does not hold, then there exists another optimal $\hat{\bm{\theta}}$ on $\S$ such that:
\begin{align}
    \mathcal{L}(\hat{\bm{\theta}},\S)<\mathcal{L}(\bm{\theta}^{*}(\N),\S).
\end{align}
That means there exists some parameter $\hat{\bm{\theta}}$ that can achieve a lower loss than $\bm{\theta}^{*}(\N)$ on the subset $\S$. Let $c={\sum\limits_{n\in\mathcal{S}}g_n(\bx, \bm{\theta}^{*}(\N))}\big/{\sum\limits_{n\in\mathcal{S}}g_n(\bx, \hat{\bm{\theta}})}$, by invoking the Universal Approximation Theorem \citep{hornik1989multilayer}, there exists parameters $\bm{\theta}'$ such that:
\begin{align}
g_n(\bx, \bm{\theta}') = \begin{cases}
    cg_n(\bx, \hat{\bm{\theta}}), & n\in\mathcal{S}, \\
    g_n(\bx, \bm{\theta}^{*}(\N)), & n\in\mathcal{N}\setminus\mathcal{S},
\end{cases}
\end{align}
which ensures $\sum\limits_{n\in\mathcal{S}}g_n(\bx, \bm{\theta}')=\sum\limits_{n\in\mathcal{S}}g_n(\bx, \bm{\theta}^{*}(\N))$. Thus, for all $\bx\in\mathcal{X}$, we have
\begin{equation}
\begin{aligned}
    \frac{cg_n(\bx, \hat{\bm{\theta}})}{\sum\limits_{n'\in \S}cg_{n'}(\bx, \hat{\bm{\theta}})}=\frac{g_n(\bx, \hat{\bm{\theta}})}{\sum\limits_{n'\in \S}g_{n'}(\bx, \hat{\bm{\theta}})}.
\end{aligned}
\end{equation}
From the definition of $\mathcal{L}(\bm{\theta},\S)$, we have:
\begin{equation}
\begin{aligned}
    &  \mathcal{L}(\hat{\bm{\theta}},\S)=  \mathcal{L}(\bm{\theta}',\S)\\
    &= \mathbb{E}_{\bx\in\mathcal{X}}\left[ \sum_{t=1}^{T} -\log\left(\sum\limits_{n\in \S}\frac{g_{n}(\bx, \bm{\theta}')}{\sum\limits_{n'\in \S}g_{n'}(\bx, \bm{\theta}')} f_{n,y_t}(h(\bx,t)) \right)\right] \\
    &= \mathbb{E}_{\bx\in\mathcal{X}}\left[ \sum_{t=1}^{T} -\left(\log\left(   \sum\limits_{n\in \S}g_{n}(\bx, \bm{\theta}') f_{n,y_t}(h(\bx,t))\right) - \right.\right. \\ 
    & \quad\quad\quad \left.\left. \log\left(\sum\limits_{n'\in \S}g_{n'}(\bx, \bm{\theta}')\right)\right)\right]\\
    &< \mathcal{L}(\bm{\theta}^{*}(\N), \S) \\
    &= \mathbb{E}_{\bx\in\mathcal{X}}\left[ \sum_{t=1}^{T} -\log\left(\sum\limits_{n\in \S}\frac{g_n(\bx, \bm{\theta}^{*}(\N))}{\sum\limits_{n'\in \S}g_{n'}(\bx, \bm{\theta}^{*}(\N))} f_{n,y_t}(h(\bx,t))\right)\right] \\
    &= \mathbb{E}_{\bx\in\mathcal{X}}\left[ \sum_{t=1}^{T} -\left(\log\left(   \sum\limits_{n\in \S}g_n(\bx, \bm{\theta}^{*}(\N)) f_{n,y_t}(h(\bx,t))\right)\right.\right. \\
    & \quad\quad\quad \left.\left. - \log\left(\sum\limits_{n'\in \S}g_{n'}(\bx, \bm{\theta}^{*}(\N))\right)\right)\right],
\end{aligned}
\end{equation}
which means
\begin{equation}
    \begin{aligned} \label{ineq1}
    & \mathbb{E}_{\bx\in\mathcal{X}}\left[ \sum_{t=1}^{T} -\log\left(\sum\limits_{n\in \S}g_n(\bx, \hat{\bm{\theta}}) f_{n,y_t}(h(\bx,t))\right) \right] \\
    &< \mathbb{E}_{\bx\in\mathcal{X}}\left[ \sum_{t=1}^{T} -\log\left(\sum\limits_{n\in \S}g_n(\bx, \bm{\theta}^{*}(\N)) f_{n,y_t}(h(\bx,t))\right) \right].\\
    \end{aligned}
\end{equation}
From the optimality of $\hat{\bm{\theta}}$, we have:
\begin{equation}
\begin{aligned}
    & \mathcal{L}(\hat{\bm{\theta}},\S) \\
    &= \mathbb{E}_{\bx\in\mathcal{X}} \left[ \sum_{t=1}^{T} -\log\left(\sum\limits_{n\in \N}\frac{g_n(\bx, \bm{\theta}^{*}(\N))}{\sum\limits_{n'\in \N}g_{n'}(\bx, \bm{\theta}^{*}(\N))} f_{n,y_t}(h(\bx,t))\right) \right]\\
    &= \mathbb{E}_{\bx\in\mathcal{X}} \left[ \sum_{t=1}^{T} -\log\left(  \sum\limits_{n\in \S}g_n(\bx, \bm{\theta}^{*}(\N)) f_{n,y_t}(h(\bx,t)) \right.\right. \\ 
    & \left.\left. \quad\quad\quad\quad + \sum\limits_{n\in \N \setminus \S}g_n(\bx, \bm{\theta}^{*}(\N)) f_{n,y_t}(h(\bx,t))\right)\right. \\
    & \left. \quad\quad\quad\quad + \log\left(\sum\limits_{n'\in \N}g_{n'}(\bx, \bm{\theta}^{*}(\N))\right)\right] \\
    & \leq \mathcal{L}(\bm{\theta}',\N) \\
    &= \mathbb{E}_{\bx\in\mathcal{X}} \left[ \sum_{t=1}^{T} -\log\left({\sum\limits_{n\in \N}g_n(\bx, \hat{\bm{\theta}})f_{n,y_t}(h(\bx,t))}\Big/ \right.\right.\\
    &\quad\quad \left.\left. \left({\sum\limits_{n'\in \S}g_{n'}(\bx, \hat{\bm{\theta}})+\sum\limits_{n'\in \N\setminus\S}g_{n'}(\bx, \bm{\theta}^{*}(\N))} \right)\right)\right] \\
    &= \mathbb{E}_{\bx\in\mathcal{X}} \left[ \sum_{t=1}^{T} - \log\left(   \sum\limits_{n\in \S}g_n(\bx, \hat{\bm{\theta}}) f_{n,y_t}(h(\bx,t)) \right.\right. \\
    & \quad\quad \left.\left. + \sum\limits_{n\in \N \setminus \S}g_n(\bx, \bm{\theta}^{*}(\N)) f_{n,y_t}(h(\bx,t))\right)\right.\\
    &\left.+ \log\left(\sum\limits_{n'\in \S}g_{n'}(\bx, \hat{\bm{\theta}})+\sum\limits_{n'\in \N\setminus\S}g_{n'}(\bx, \bm{\theta}^*(\N))\right)\right]
\end{aligned}
\end{equation}
Thus, we have
\begin{equation}
\begin{aligned} \label{ineq2}
    &\mathbb{E}_{\bx\in\mathcal{X}} \left[ \sum_{t=1}^{T} -\log\left(   \sum\limits_{n\in \S}g_n(\bx, \bm{\theta}^{*}(\N)) f_{n,y_t}(h(\bx,t)) \right) \right.\\
    & \quad\quad + \left. \log\left(\sum\limits_{n'\in \N}g_{n'}(\bx, \bm{\theta}^{*}(\N))\right) \right]\\
    &\leq \mathbb{E}_{\bx\in\mathcal{X}} \left[ \sum_{t=1}^{T} -\log\left(\sum\limits_{n\in \S}g_n(\bx, \hat{\bm{\theta}}) f_{n,y_t}(h(\bx,t)) \right)\right. \\ 
    &+ \left. \log\left(\sum\limits_{n'\in \S}g_{n'}(\bx, \hat{\bm{\theta}})+\sum\limits_{n'\in \N\setminus\S}g_{n'}(\bx, \bm{\theta}^{*}(\N))\right)\right]
\end{aligned}
\end{equation}
From  \eqref{ineq1} and \eqref{ineq2}, we have
\begin{equation}
\begin{aligned}
    & \mathbb{E}_{\bx\in\mathcal{X}} \left[ \sum_{t=1}^{T} -\log\left(\frac{\sum\limits_{n'\in \S}g_{n'}(\bx, \hat{\bm{\theta}}) f_{n,y_t}(h(\bx,t))}{\sum\limits_{n'\in \S}g_{n'}(\bx, \bm{\theta}^{*}(\N)) f_{n,y_t}(h(\bx,t))}\right) \right]\\
    % &< \mathbb{E}_{\bx\in\mathcal{X}} \left[\sum_{t=1}^{T} -\log\left(\frac{\sum\limits_{n'\in \S}g_{n'}(\bx, \hat{\bm{\theta}}) f_{n,y_t}(h(\bx,t))+\sum\limits_{n'\in \N\setminus\S}g_{n'}(\bx, \bm{\theta}^{*}(\N)) f_{n,y_t}(h(\bx,t))}{\sum\limits_{n'\in \S}g_{n'}(\bx, \bm{\theta}^{*}(\N)) f_{n,y_t}(h(\bx,t))+\sum\limits_{n'\in \N\setminus\S}g_{n'}(\bx, \bm{\theta}^{*}(\N)) f_{n,y_t}(h(\bx,t))}\right) \right],
    &< \mathbb{E}_{\bx\in\mathcal{X}} \left[\sum_{t=1}^{T} -\log\left(\sum\limits_{n'\in \S}g_{n'}(\bx, \hat{\bm{\theta}}) f_{n,y_t}(h(\bx,t)) \right.\right.\\
    & \quad\quad\quad \left.\left.\left. +\sum\limits_{n'\in \N\setminus\S}g_{n'}(\bx, \bm{\theta}^{*}(\N)) f_{n,y_t}(h(\bx,t))\right) \right.\right.\\
    & \quad \left.\left. +\log\left(\sum\limits_{n'\in \S}g_{n'}(\bx, \bm{\theta}^{*}(\N)) f_{n,y_t}(h(\bx,t)) \right.\right.\right.\\
    & \quad\quad\quad \left.\left. +\sum\limits_{n'\in \N\setminus\S}g_{n'}(\bx, \bm{\theta}^{*}(\N)) f_{n,y_t}(h(\bx,t))\right) \right],
\end{aligned}
\end{equation}
which contradicts with
\begin{equation}
\begin{aligned}
    & \mathbb{E}_{\bx\in\mathcal{X}} \left[ \sum_{t=1}^{T} -\log\left(\sum\limits_{n'\in \S}g_{n'}(\bx, \hat{\bm{\theta}}) f_{n,y_t}(h(\bx,t))\right)\right] \\
    & \leq \mathbb{E}_{\bx\in\mathcal{X}} \left[ \sum_{t=1}^{T} -\log\left(\sum\limits_{n'\in \S}g_{n'}(\bx, \bm{\theta}^{*}(\N)) f_{n,y_t}(h(\bx,t))\right)\right].
\end{aligned}
\end{equation}
\end{proof}

\begin{Theorem}[Optimality for subset]
Let $\bm{\theta}^{*}(\N)=\arg\min\limits_{\bm{\theta}}\mathcal{L}(\boldsymbol{\theta}, \N)$ denote the optimal parameters for the gating network given experts set $\mathcal{N}$. Given the following loss function:
\begin{equation}
\begin{aligned}
    & \mathcal{L}(\bx, \bm{\omega}, \S) = -\sum_{t=1}^{T} \log\left(\sum\limits_{n\in \mathcal{S}}  \omega_n f_{n}(h(\bx,t))\right),
\end{aligned}
\end{equation}
Then we have
\begin{align}
    \bm{\omega}(\bx, \bm{\theta}^{*}(\N), \S)=\min_{\bm{w}} \mathcal{L}(\bx, \bm{\omega}, \S).
\end{align}
\end{Theorem}
\begin{proof}

From the definition, we know that:
\begin{equation}
    \begin{aligned}
    \mathcal{L}(\boldsymbol{\theta}^*,\N) 
    &= \min_{\bm{\theta}} \mathcal{L}(\boldsymbol{\theta},\N) \\
    &= \mathop{\mathbb{E}}\limits_{\bx\sim \mathcal{P}(\mathcal{X})}\left[-\sum_{t=1}^{T} \log\left(\sum\limits_{n\in \N}  \omega_n(\bx, \bm{\theta}^{*}(\N), \N) f_{n}(h(\bx,t))\right)\right] \\
    &= 
    \end{aligned}
\end{equation}
If we consider a subset $\S\subseteq N$, we have:
\begin{align}
\mathcal{L}(\boldsymbol{\theta}^*,\S) = \mathop{\mathbb{E}}\limits_{\bx\sim \mathcal{P}(\mathcal{X})}\left[-\sum_{t=1}^{T} \log\left(\sum\limits_{n\in \S}  \omega_n(\bx, \bm{\theta}^{*}(\N), \S) f_{n}(h(\bx,t))\right)\right],
\end{align}
Suppose there exists $\hat{\bm{\theta}}$ that:
\begin{equation}
    \mathcal{L}(\hat{\bm{\theta}},\S) = \min_{{\bm{\theta}}} \mathcal{L}(\boldsymbol{\theta},\S)
\end{equation}
and:
\begin{equation}
\begin{aligned}
    \mathcal{L}(\hat{\bm{\theta}},\S) < \mathcal{L}(\boldsymbol{\theta}^*,\S), \\
    \mathcal{L}(\hat{\bm{\theta}},\N) > \mathcal{L}(\boldsymbol{\theta}^*,\N)
\end{aligned}
\end{equation}
which means:
\begin{equation}
\begin{aligned}
    \mathop{\mathbb{E}}\limits_{\bx\sim \mathcal{P}(\mathcal{X})}\left[-\sum_{t=1}^{T} \log\left(\sum\limits_{n\in \S}  \omega_n(\bx, \hat{\bm{\theta}}, \S) f_{n}(h(\bx,t))\right)\right] \\
    < \mathop{\mathbb{E}}\limits_{\bx\sim \mathcal{P}(\mathcal{X})}\left[-\sum_{t=1}^{T} \log\left(\sum\limits_{n\in \S}  \omega_n(\bx, \bm{\theta}^{*}(\N), \S) f_{n}(h(\bx,t))\right)\right]
\end{aligned}
\end{equation}
and:
\begin{equation}
    \label{Eq: inequality on N}
    \begin{aligned}
        &\mathop{\mathbb{E}}\limits_{\bx\sim \mathcal{P}(\mathcal{X})}\left[-\sum_{t=1}^{T} \log\left(\sum\limits_{n\in \N}  \omega_n(\bx, \bm{\theta}^{*}(\N), \N) f_{n}(h(\bx,t))\right)\right] \\
        % &= \mathop{\mathbb{E}}\limits_{\bx\sim \mathcal{P}(\mathcal{X})}\left[-\sum_{t=1}^{T} \log\left(\sum\limits_{n\in \S}  \omega_n(\bx, \bm{\theta}^{*}(\N), \N) f_{n}(h(\bx,t)+\sum\limits_{n\in {\N/\S}}  \omega_n(\bx, \bm{\theta}^{*}(\N), \N) f_{n}(h(\bx,t))\right)\right] \\
        &< \mathop{\mathbb{E}}\limits_{\bx\sim \mathcal{P}(\mathcal{X})}\left[-\sum_{t=1}^{T} \log\left(\sum\limits_{n\in \S}  \omega_n(\bx, \hat{\bm{\theta}}, \N) f_{n}(h(\bx,t))\right)\right]
        % &= \mathop{\mathbb{E}}\limits_{\bx\sim \mathcal{P}(\mathcal{X})}\left[-\sum_{t=1}^{T} \log\left(\sum\limits_{n\in \S}  \omega_n(\bx, \hat{\bm{\theta}}, \N) f_{n}(h(\bx,t)+\sum\limits_{n\in {\N/\S}}  \omega_n(\bx, \hat{\bm{\theta}}, \N) f_{n}(h(\bx,t))\right)\right]
    \end{aligned}
\end{equation}
Note that for any $\bm{\theta}$, we have:
\begin{equation}
    \begin{aligned}
        \omega_n(\bx, {\bm{\theta}}, \S) &= \frac{\exp(g_n(\bx,\boldsymbol{\theta}))}{\sum\limits_{n'\in \mathcal{S}}\exp(g_{n'}(\bx,\boldsymbol{\theta}))} \\
        &> \frac{\exp(g_n(\bx,\boldsymbol{\theta}))}{\sum\limits_{n'\in \mathcal{N}}\exp(g_{n'}(\bx,\boldsymbol{\theta}))} = \omega_n(\bx, {\bm{\theta}}, \N)
    \end{aligned}
\end{equation}
\end{proof}
Therefore, Eq. \ref{Eq: inequality on N} can be rewritten as:
\begin{equation}
    \begin{aligned}
        &\mathop{\mathbb{E}}\limits_{\bx\sim \mathcal{P}(\mathcal{X})}\left[-\sum_{t=1}^{T} \log\left(\sum\limits_{n\in \N}  \omega_n(\bx, \bm{\theta}^{*}(\N), \S) f_{n}(h(\bx,t))\right)\right] \\
        &< \mathop{\mathbb{E}}\limits_{\bx\sim \mathcal{P}(\mathcal{X})}\left[-\sum_{t=1}^{T} \log\left(\sum\limits_{n\in \S}  \omega_n(\bx, \hat{\bm{\theta}}, \N) f_{n}(h(\bx,t))\right)\right]
    \end{aligned}
\end{equation}

\bibliographystyle{plainnat}
\bibliography{ref}